\documentclass[letterpaper, 10 pt, conference]{ieeeconf} 
\IEEEoverridecommandlockouts

\usepackage{enumitem}
\setlist{topsep=0pt, leftmargin=*}
\usepackage{algorithm}
\usepackage{algorithmicx}

\usepackage[noend]{algpseudocode}
\usepackage{amsmath, amsfonts, bm, amssymb, tabularx}
\usepackage{latexsym, mathtools}
\usepackage{amsthm}
\usepackage{ifthen}
\usepackage{hyperref}\hypersetup{colorlinks=true, unicode=true, linkcolor=[rgb]{0.10,0.05,0.67}, citecolor=[rgb]{0.10,0.05,0.67}, filecolor=[rgb]{0.10,0.05,0.67}, urlcolor=[rgb]{0.10,0.05,0.67}}
\usepackage{Shorthands}
\usepackage{cleveref}
\newcommand{\citep}[1]{\cite{#1}}
\usepackage{stackengine}

\usepackage[numbers, sort&compress]{natbib}
\allowdisplaybreaks

\title{\LARGE \bf
Statistical Efficiency of Single- and Multi-step Models for Forecasting and Control
}
\usepackage{subcaption}
\author{
Anne Somalwar$^{1}$, Bruce D. Lee, George J. Pappas, Nikolai Matni% <-this % stops a space
\thanks{Anne Somalwar, George J. Pappas, and Nikolai Matni are with the University of Pennsylvania. Emails : \tt\small
\{somalwar, nmatni, pappasg\}@seas.upenn.edu}
\thanks{Bruce D. Lee is with the ETH AI Center. Email: \tt\small bruce.lee@ai.ethz.ch}}

\begin{document}

\twocolumn

\maketitle
\thispagestyle{empty}
\pagestyle{empty}
%%%%%%%%%%%%%%%%%%%%%%%%%%%%%%%%%%%%%%%%%%%%%%%%%%%%%%%%%%%%%%%%%%%%%%%%%%%%%%%%
\begin{abstract}
Compounding error, where small prediction mistakes accumulate over time, presents a major challenge in learning-based control. A common remedy is to train multi-step predictors directly instead of rolling out single-step models. However, it is unclear when the benefits of multi-step predictors outweigh the difficulty of learning a more complex model. We provide the first quantitative analysis of this trade-off for linear dynamical systems. We study three predictor classes: (i) single step models, (ii) multi-step models, and (iii) single step models trained with multi-step losses. We show that when the model class is well-specified and accurately captures the system dynamics, single-step models achieve the lowest asymptotic prediction error. On the other hand, when the model class is misspecified due to partial observability, direct multi-step predictors can significantly reduce bias and improve accuracy. We provide theoretical and empirical evidence that these trade-offs persist when predictors are used in closed-loop control. 
\end{abstract}
\section*{Introduction}
\label{sec:introduction}

A typical approach to time series forecasting is to fit a one-step ahead prediction model and apply it recursively to obtain predictions over multiple time steps. In doing so, small errors may compound over time, leading to poor long-horizon prediction. This issue can make the application of such single-step models to long-term planning and controller design challenging.

By directly training multi-step models to predict over longer horizons, the issue of compounding error can be mitigated. The main drawback of doing so is that the number of parameters for a direct multi-step predictor scales with the prediction horizon, thus potentially requiring more data than a single-step predictor to achieve a desired prediction performance. While this tradeoff between prediction horizon accuracy and data requirements is broadly known to exist, 
 it is primarily studied from an empirical perspective \cite{lambert2021learning, lambert2022investigating}. We therefore lack principled guidance for exactly when direct multi-step prediction should be preferred over autoregressive rollout of single-step models. Motivated by this challenge, we provide a rigorous comparison of the sample efficiency of learning multi-step predictors with that of learning single-step predictors in the %simplified yet representative
 setting of a linear dynamical system.

\subsection{Related Work}

\paragraph{Multi-step Identification}
The goal of system identification is to use data to learn a model that can be used for forecasting or control \cite{ljung1998system}. 
To this end, one typically wants to select a model from the hypothesis class that minimizes the simulation error, i.e., the cumulative prediction error over all future time steps. Due to the computational challenge of doing so, it is much more common to instead learn a model which minimizes the single-step prediction error and apply it autoregressively \cite{ljung1998system, farina2008some}. However, such approaches tend to generalize poorly if the underlying data generating process does not belong to the hypothesis class. 
%This has motivated the application of approaches from imitation learning to perform data augmentation that refines a model minimizing the single-step prediction error, leading to improved empirical performance 
This has motivated the application of algorithms such as Data as Demonstrator (DaD) which use approaches from imitation learning to train a predictor to self-correct its prediction error at each step, leading to improved empirical performance \cite{venkatraman2015improving, venkatraman2017improved}. 

Direct learning of prediction models for each time-step in the prediction horizon partially bypasses the issue of compounding error, and has proven successful for both model predictive control \cite{terzi2018learning, balim2024stochastic, kohler2022state} and value approximation in MDPs \cite{asadi2019combating}. Empirical studies of single-step and multi-step dynamics models learned with various neural network architectures have also been conducted. Namely, \cite{lambert2022investigating} investigates the performance of recursive application of single-step models parameterized by neural networks in a handful of examples and characterize circumstances which such models may worsen the effects of compounding error. In \cite{chandra2021evaluation}, the authors give a comparison of various deep learning architectures for predicting multiple steps of a time series and show the efficacy of  bidirectional and encoder-decoder LSTM networks.

These empirical studies underscore the importance of careful consideration when designing a model for multi-step prediction. Our work studies this question in stylized settings, allowing us to clearly demonstrate the potential drawbacks and benefits of direct multi-step prediction as compared to more traditional single-step approaches. 

\paragraph{Learning-Enabled Control} 

While the issue of compounding error has long been studied in system identification and control, it has resurfaced as a prominent issue in the learning community, exacerbated by the use of neural networks as function approximators. In model-based reinforcement learning, it has been observed that synthesized controllers may exploit compounding errors of the learned model \cite{chua2018deep,levine2020offline}, motivating numerous heuristics for accounting for the model error during policy synthesis \cite{janner2019trust, yu2020mopo, yu2021combo, kidambi2020morel}. In \cite{lambert2021learning}, the authors instead propose learning a direct multi-step predictor parametrized by a low dimensional decision variable which may be optimized online. 

The issue of a mismatch between the hypothesis class and the underlying data generating process also poses a challenge for behavior cloning. For example, incorrectly assuming that the demonstrator is Markovian can lead to a policy that deviates substantially from the demonstrator \cite{chi2023diffusion, block2023provable}. This can be remedied in part by replacing the standard behavior cloning objective with an objective that predicts the expert actions for multiple timesteps to maintain temporal consistency, resulting in so-called \emph{action-chunking} based approaches~\cite{block2023provable, zhao2023learning}. We draw inspiration from these studies, and consider instances of linear systems with either Markovian or non-Markovian observations to compare direct multi-step prediction and autoregressive evaluation of single-step predictors.

\subsection{Contributions}

Our prior work \cite{CDCPaper} provided a quantitative comparison of the multi-step prediction error incurred by a directly learned multi-step predictor and autoregressive evaluation of a single-step predictor. 
Here, we extend this work by: 
(i) considering a third class of predictors, specifically a single-step predictor trained by minimizing a multi-step loss, 
(ii) analyzing the resulting closed-loop control performance, 
and (iii) providing complete proofs.
In particular:
\begin{itemize}[noitemsep,nolistsep]
    \item We provide an asymptotic characterization of the multi-step prediction error for the three methods in the setting of a fully observed dynamical system. Our results show that for stable systems with a small spectral radius, the prediction error of autoregressive evaluation for a single-step predictor decays significantly faster with increasing data than that of the other methods.
    \item We characterize the prediction error for the three methods applied to a partially observed dynamical system that is incorrectly assumed to be fully observed by the hypothesis class, thus addressing the issue of trajectory prediction under misspecification due to an unjustified Markovian assumption. These results demonstrate that multi-step predictors may enjoy significantly lower bias in the face of model misspecification.
    \item We analyze predictor performance in a closed loop control setting. Our results show that, in the case of a fully observed dynamical system (well-specified setting), autoregressive evaluation of a single-step predictor fitted with a single step loss often yields lower control cost than autoregressive evaluation of a single step predictor fitted with a multistep loss.
    \item We conduct numerical experiments that exemplify the above results. Code to reproduce these experiments is available. \footnote{Code to reproduce experiments: https://github.com/annesomalwar/Multistep-Prediction.git}
\end{itemize}

Our results, while limited to stylized settings, capture key properties of real systems---such as partial observability and misspecification---and thus provide useful guidance to practitioners navigating the complex design space of data-driven multi-step predictors.\\ 

\noindent
\textbf{Notation}: $\mathcal{N}(\mu, \sigma^2)$ denotes a normal distribution with mean $\mu$ and variance $\sigma^2$, $\probconv$ denotes convergence in probability, $\rho(\cdot)$ denotes the spectral radius of a matrix, $\norm{\cdot}$ denotes the vector Euclidean norm, $\norm{\cdot}_F$ denotes the matrix Frobenius norm, $\norm{\cdot}_\star$ denotes the matrix nuclear norm, $\VEC(\cdot)$ denotes the vectorization of a matrix, and $\otimes$ denotes the Kronecker product. For a random sequence $X_n$, we use $X_n = {\cal O}_p\paren{a_n}$ to denote that the set of values $X_n/a_n$ is stochastically bounded. We use the notation $X_{a:b}$ to denote the sequence $X_n$ for $n=a$ to $n=b$. For functions $f,g$ of $x$, we use the notation $f(x) = o(g(x))$ to denote that $\lim_{x \rightarrow \infty} \frac{f(x)}{g(x)} = 0$.

% \noindent
% \sloppy
% Extended manuscript with full proofs can be found at: \\
% \noindent
% \url{https://drive.google.com/file/d/1mYoQB53ne_NRed5_atEqnji77RpfUYAa/view?usp=sharing}
% \url{https://arxiv.org/abs/2504.01766}
\section{Problem Formulation}

Consider the linear time invariant dynamical system 
\begin{equation}
\label{eq: dynamics}
\begin{aligned}
    x_{t+1} &= Ax_t + Bu_t + B_w w_t &&\quad t \in \mathbb{Z}^+\\
    y_t &= C x_t + D_v v_t,  &&\quad t \in \mathbb{Z}^+
\end{aligned}
\end{equation}
with state $x_t \in \R^{\dx}$, input $u_t \in \R^{\du}$, observation $y_t\in R^{\dy}$, process noise $w_t \overset{\iid}{\sim} \calN(0,I_{d_x})$, sensor noise $v_t  \overset{\iid}{\sim} \calN(0,I_{d_y})$, and initial condition $x_0 = 0$. We assume that $(A,C)$ is observable, that $(A, \bmat{B, B_w})$ is controllable, $\rho(A) < 1$\footnote{The assumption $\rho(A)<1$ guarantees stationarity and ergodicity of the process $\{x_t\}$.}, and that the control inputs are selected randomly as $u_t\overset{\iid}{\sim}\calN(0,I_{d_u})$.  %In particular, the distinction between the control input and noise in this setting is that the control input is known. 

We assume that the dynamics \eqref{eq: dynamics} are unknown, and our goal is to learn a predictor that forecasts a horizon $H$ of future observations using past observations. To this end, we suppose that we are given a dataset $\calD_N = \curly{(y_t, u_t)}_{t=1}^N$ %{\color{red} since we're considering both one step and multistep predictors here should the data just be $(y_t,u_t)$? or is this standard?} 
collected from a training rollout of \eqref{eq: dynamics} which will be used to determine a function $\hat f_H$ belonging to a hypothesis class $\calF_H$. This function will be used to predict $y_{t+1:t+H}$ given $y_{1:t}$ and $u_{1:t+H-1}$. 

The quality of the learned function will be measured by the loss
\begin{align}
    \label{eq: loss}
    L(\hat f_H) \triangleq  \bar \E \norm{y_{t+1:t+H} - \hat f_H(y_{1:t}, u_{1:t+H-1})}^2,
\end{align}
where the operator $\bar \E$ is defined as
\begin{align*}
    \bar \E[f(t)] \triangleq \lim_{T\to\infty} \mathbf{E} \frac{1}{T}\sum_{t=0}^T f(t), 
\end{align*}
and the expectation is taken over an evaluation rollout of system \eqref{eq: dynamics} that is independent of the dataset $\calD_N$. By ergodicity of the process \eqref{eq: dynamics}, this is equivalent to taking an expectation under the steady state distribution for the system. 

To provide rigorous understanding of situations where multi-step prediction does or does not help, we consider a simplified setting in which the hypothesis class consists of static linear predictors, i.e., a function $f_H \in \calF_H$ given by \begin{equation}\label{eq:gen-pred}
f_H(y_{1:t}, u_{1:t+H-1}) = G \bmat{y_t \\ u_{t:t+H-1}}, 
\end{equation}
for a matrix $G \in S \subseteq \R^{H \dy \times (\dy + H\du)}$.  Here the subspace $S$ encodes whether we are fitting a multi-step or single-step model: we provide explicit parameterizations for these model-classes in the next subsections. Our restriction to static linear predictors of the form~\eqref{eq:gen-pred} assumes that the observation sequence is Markovian, i.e., that a history of observations is unnecessary to predict the future trajectory.  In the sequel, we slightly abuse notation and denote the loss~\eqref{eq: loss} incurred by a predictor~\eqref{eq:gen-pred} defined by matrix $\hat G$ by $L(\hat G).$ 

We consider two settings: one where the Markovian assumption is justified ($C=I$ and $D_v = 0$), resulting in a well-specified problem, and one where it is not justified ($C \neq I$ or $D_v D_v^\top \succ 0$), resulting in a misspecified problem. 
In these two settings, we compare the $H$ step prediction error~\eqref{eq: loss} incurred by a learned single-step model rolled out for $H$ timesteps to that incurred by a directly learned $H$-step model. 

Throughout, we use ``misspecification`` specifically to refer to violations of the Markov assumption induced by partial observability (and observation noise), while ``well-specified`` refers to the fully observed state setting.

\subsection{Single-step Predictors}
\label{s:singlestep}
The single-step approach first solves
\begin{align}
    \label{eq: single-step LS}
    \bmat{\hat G_y & \hat G_u} = \argmin_{\substack{G_y \in \R^{\dy \times \dy} \\ G_u \in \R^{\dy \times \du}} } \sum_{t=1}^{N-1} \norm{y_{t+1} - \bmat{G_y & G_u} \bmat{y_t \\ u_t}}^2.
\end{align}

Using this model, one can predict $y_{t+1:t+H}$ by rolling out $\bmat{\hat G_y & \hat G_u}$ autoregressively:
\begin{align}
    \label{eq: single-step rollout}
    \hat y_{t+1} =&  \bmat{\hat G_y & \hat G_u} \bmat{y_t \notag \\ u_t} \\
    \hat y_{t+2} =& \bmat{\hat G_y & \hat G_u} \bmat{\hat y_{t+1} \\ u_{t+1}} \notag \\ 
    &\vdots \notag \\ 
    \hat y_{t+H} =& \bmat{\hat G_y & \hat G_u} \bmat{\hat y_{t+H-1} \\ u_{t+H-1}}. 
\end{align}
The resulting $H$-step predictor can be composed to form a direct mapping from the data to the predicted trajectory as 
\begin{align}
    \label{eq: single-step rolled out G}
    \hat G^{SS}_N = \bmat{\hat G_y & \hat G_u & 0 & \dots & 0 \\
                   \hat G_y^2 & \hat G_y \hat G_u & \hat G_u   &\dots & 0 \\
                   \vdots  &&& \ddots \\
                   \hat G_y^{H} &   \hat G_y^{H-1}\hat G_u &  \hat G_y^{H-2}\hat G_u &\dots & \hat G_u \\}.
\end{align}
%The subscript on $\hat G$ emphasizes the dependence on the amount of data $N$, which we drop going forward for notational clarity.% in subsequent sections and refer to $\hat G_N$ as $\hat G$. 
As past predictions become part of the regressor for future predictions, this approach often suffers from compounding error. 

\subsection{Multi-step Predictors}
\label{s: multi-step}
The issue of compounding error from autoregressive roll-out of a single-step model motivates direct multi-step approaches which directly minimize the $H$ step prediction error:
\begin{align}
    \label{eq: multistep LS}
    \hat G^{MS}_N = \argmin_{G \in S} \sum_{t=1}^{N-H} \norm{y_{t+1:t+H} - G \bmat{y_t \\ u_{t:t+H-1}}}^2,
\end{align}
for $S \subseteq \R^{H\dy \times (\dy + H\du)}$. %(Once again, we will omit the subscript and refer to $\hat G_N$ as $\hat G$ in subsequent sections.  
We consider the function class which fits $H$ distinct predictors, one for each step in the prediction horizon. This amounts to setting $S = \R^{H \dy \times (\dy + H \du)}$.\footnote{One could impose the causality structure, i.e. that $S$ has a triangular structure. We refrain from doing so for simplicity, and due to the fact that future inputs are independent of the past.}

There is a tradeoff induced by fitting multi-step predictors rather than single-step predictors. In particular, the single-step predictor is subject to compounding error, while the complexity of the above identification problem increases for longer horizons. 

\subsection{Intermediate Formulations}
\label{s: intermediate}

\stackMath
Rather than fitting independent predictors for every timestep, one can instead formulate a hypothesis class for multi-step prediction with lower complexity. In particular, one could impose additional structure on $S$. For example, let 
\begin{align}
    \label{eq: multistep LS}
    \hat G^{I}_N \in \argmin_{G \in S} \sum_{t=1}^{N-H} \norm{y_{t+1:t+H} - G \bmat{y_t \\ u_{t:t+H-1}}}^2
\end{align}
where
\begin{align}\label{eq: multistep loss} S = \curly{\bmat{ G_y &  \!\!\!\!\!G_u & 0 & \!\!\!\!\!\dots & \!\!\!\!\!0 \\
                    G_y^2 & \!\!\!\!\!G_y G_u & \!\!\!\!\!G_u   &\!\!\!\!\!\dots & \!\!\!\!\!0 \\
                   \vdots  &&& \ddots \\
                   G_y^{H} & \!\!\!\!\!  G_y^{H-1} G_u & \!\!\!\!\! G_y^{H-2} G_u &\!\!\!\!\!\dots & \!\!\!\!\!G_u \\} \Bigg \vert \raisebox{-8pt}{\stackon[1pt]{G_u \in \R^{\dy \times \du}}{G_y \in \R^{\dy \times \dy}}}}. \end{align} 
This consists of functions which take the form of a single-step predictor that is applied auto-regressively. In contrast to the single-step approach of \Cref{s:singlestep}, solving \eqref{eq: multistep LS} with this choice of $S$ consists of a multi-step loss function for a class of single-step predictors, a common approach to mitigate the compounding error issue without increasing the number of parameters that must be learned \citep{farina2008some}.  

We compare single-step, multi-step, and the intermediate predictors described above in the two aforementioned settings: a system with Markovian observations, and a system with non-Markovian observations. Due to the Markovian assumption for the identification problem, these cases serve as instances where the identification problem is well-specified and misspecified, respectively.

%and \eqref{eq: shared representation} 
\section{Well-Specified Setting}
\label{sec: well-specified}

In this section, we study the well-specified setting in which the Markovian assumption is valid. In particular, we restrict system \eqref{eq: dynamics} to be a fully observed system by assuming that $C = I$ and $D_v = 0$ so $y_t = x_t$ for all $t$.  

To compare the three approaches in this setting, we first observe that the predictors $\hat f_H$ are defined in terms of a linear map $\hat G$ applied to the vector $\bmat{x_t^\top & u_{t:t+H-1}^\top}^\top$. Therefore the loss \eqref{eq: loss} may be written
\begin{align}
    \label{eq: loss expanded}
    L(\hat G) = \bar \E \norm{x_{t+1:t+H} - \hat G \bmat{x_{t} \\ u_{t:t+H-1}}}^2. 
\end{align}
Rolling out the dynamics, we find that 
\begin{align*}
    x_{t+1:t+H} = G^\star  \bmat{x_t \\ u_{t:t+H-1}} + \Gamma_w w_{t:t+H-1},
\end{align*}
where 
\begin{align*}
    %\Gamma_X &= \bmat{A \\ A^2 \\ \vdots \\ A^H}, \Gamma_U = \bmat{B \\  AB & B \\ \vdots && \ddots \\ A^{H-1}B && \dots & B} \\
    G^\star &=  \bmat{A & B \\ A^2&  AB & B \\ &\vdots && \ddots \\ A^H & A^{H-1}B && \dots & B},
\end{align*}
and
\begin{align*}
    \Gamma_w &= \bmat{B_w \\  AB_w & B_w \\ \vdots && \ddots \\ A^{H-1} B_w & \dots && B_w}.
\end{align*} 

Then expanding $x_{t+1:t+H}$ in equation \eqref{eq: loss expanded}, and using the independence of $w_{t:t+H-1}$ from $x_t$ and $u_{t:t+H-1}$, we conclude that
\begin{align*}
    L(\hat G) &= \bar\E \norm{(\hat G - G^\star) \bmat{x_t \\ u_{t:t+H-1}}}^2+\bar \E \norm{\Gamma_w w_{t:t+H-1}}^2 \\
    & \\
    &= \norm{(\hat G - G^\star) \Sigma_z^{1/2}}_F^2 + \norm{\Gamma_w}_F^2,
\end{align*}
where $\Sigma_z = \bar \E z_t z_t^\top $ is the stationary covariance for the regressor $z_t \triangleq \bmat{x_t \\ u_{t:t+H-1}}$. \footnote{By the fact that $w_t, u_t$ are i.i.d. standard normal and $(A, \bmat{B, B_w})$ is controllable, persistence of excitation holds, i.e. $\Sigma_z$ is positive definite. }
Consequently, the discrepancy between the single-step and multi-step predictors is contained in the term $\norm{(\hat G - G^\star) \Sigma_z^{1/2}}_F^2$. We study the behavior of this term asymptotically, where $\hat G$, or equivalently $\hat G_N$, is an operator learned on the dataset of size $N$.\footnote{We sometimes omit the subscript $N$ on $\hat G_N$ to ease notational burden.} In particular, we examine
\begin{align*}
    \lim_{N\to\infty} N \E\brac{\norm{(\hat G_N - G^\star) \Sigma_z^{1/2}}_F^2},
\end{align*} 
for the predictors $\hat G^{SS}_N$, $\hat G^{MS}_N$, and $\hat G^{I}_N$, where the expectation is taken over the dataset used to fit $\hat G_N$. 

The reducible error of the multi-step predictor is characterized by the following proposition. 
%We provide non-asymptotic bounds in the appendix; however, these bounds are not sufficiently sharp to distinguish between the single-step predictor and the multi-step predictor. 

\begin{proposition}[Proposition~III.1 of \cite{CDCPaper}]
    \label{prop: well specified multistep}
    The reducible asymptotic error of the multi-step predictor $\hat G^{MS}_N$ is given by 
    \begin{align*}
        &\lim_{N\to\infty} N \mathbf{E} \brac{\norm{(\hat G^{MS}_N - G^\star) \Sigma_z^{1/2}}_F^2}\\
        &=\trace\paren{\Gamma_w ((M_{MS} + H\du I_H) \otimes I_{\dx}) \Gamma_w^\top},
    \end{align*}
    where $M_{MS} \in \R^{H \times H }$ is the matrix with entry $(i,j)$ given by
    % \begin{align*}
    $
        M_{MS}^{ij} = \trace(A^{\abs{i-j}})$.
        
\end{proposition}

\begin{proof}
By the normal equations for the least-squares estimator,
\begin{align*}
\hat G_N^{\mathrm{MS}} - G^\star
=
\sum_{t=1}^{N-H+1}
\Gamma_w w_{t:t+H-1} z_t^\top
\left(\sum_{t=1}^{N-H+1} z_t z_t^\top\right)^{-1}.
\end{align*}
From a combination of Slutsky’s theorem, Birkhoff-Khinchin theorem, and Vitali's  convergence theorem, 
\begin{align*}
&\lim_{N\to\infty}
N\,\E\!\left\|(\hat G_N^{\mathrm{MS}}-G^\star)\Sigma_z^{1/2}\right\|_F^2 \\
&=
\lim_{N\to\infty}
\frac{1}{N}\,\E\Biggl\|
\sum_{t=1}^{N-H+1}
\Gamma_w w_{t:t+H-1} z_t^\top \Sigma_z^{-1/2}
\Biggr\|_F^2 .
\end{align*}
Expanding the Frobenius norm results in a double sum over time indices. The evaluation of each term, accounting for temporal dependence, is given in \Cref{lemma: 1 prop2.1}. Summing over all indices and combining terms yields the characterization stated in the proposition. 
\end{proof}

 The above result shows that the error decays asymptotically at a rate of $1/N$. The scaling is characterized by the trace expression, which represents the asymptotic covariance of the estimation error; importantly, it grows with the horizon $H$ (note the $(M_{MS} + H d_u I_H)$ term). We will contrast this with the error of the single-step predictor,  characterized below. 

% In particular, the above result shows that the error decays asymptotically at a rate of $1/N$, with the noise level defined in terms of the matrix $\Gamma_w$ a future horizon of noise to the states, the input dimension $\du$, the horizon $H$, the transition matrix $A$, and the noise covariance $B_w B_w^\top$. 

\begin{proposition}[Proposition~III.2 of \cite{CDCPaper}]
    \label{prop: single-step well specified}
    The asymptotic error of the single-step predictor $\hat G^{SS}_N$ is given by 
    \begin{align*}
        &\lim_{N\to\infty} N \mathbf{E} \brac{\norm{(\hat G^{SS}_N - G^\star) \Sigma_z^{1/2}}_F^2} \\
        % &= \sum_{i,j=1}^H \trace ((e_j e_i^T \otimes \Sigma_{Z,1}^{-1})F \Sigma_zF^T)\trace(\Gamma_w(e_je_i^T \otimes I_{\dx})\Gamma_w^T).
        &=\trace(\Gamma_w ((M_{SS} + \du I_H) \otimes I_{\dx}) \Gamma_w^T),
    \end{align*}
    % \Bruce{Is I scaled by $H$, or did I make a mistake in the multistep?}
    where $M_{SS} \in \R^{H \times H}$ is the matrix with entry $(i,j)$ given by
    \begin{align*}
        % M_{SS}^{ij} \triangleq  \trace\paren{\Sigma_x^{-1} \paren{A^{i-1} \Sigma_x (A_{i-1})^\top + \sum_{\ell=0}^{i-2} A^{\ell}BB^\top (A^\ell)^\top} (A^{j-i})^\top}
          \trace\paren{\paren{ I - \Sigma_x^{-1} \sum_{\ell=0}^{\min\curly{i,j}-2} A^{\ell}B_wB_w^\top (A^\ell)^\top} (A^{\abs{j-i}})^\top}.
    \end{align*}
\end{proposition}
\begin{proof}  
Manipulating \eqref{eq: single-step rolled out G}, we can write $\hat G_N^{SS} - G^\star$ as a sum of terms which are linear in $\bmat{\hat G_y \!-\! A\! & \!\!\hat G_u \!-\! B}$ and higher order terms. That is,
\begin{align*}
    \hat G^{SS}_N - G &= \Gamma (I_H \otimes (\bmat{\hat G_y - A & \hat G_u - B}  ))F + {\cal O}_p\paren{\frac{1}{N}}
\end{align*}
where $F$ and $\Gamma$ are functions of $A, B, H$ defined in Appendix~\ref{appendix: single step well specified}. Noting that higher order terms vanish in the limit, 
\begin{align} \label{eq: proof II.2 loss}
        &\lim_{N\to\infty} N \mathbf{E} \brac{\norm{(\hat G^{SS}_N - G^\star) \Sigma_z^{1/2}}_F^2} \notag\\
        &= \lim_{N\to\infty} N \mathbf{E} \brac{\norm{\Gamma \left(I_H \otimes \left(\bmat{\hat G_y - A & \hat G_u - B}  \right) \right)F \Sigma_z^{1/2}}_F^2} \notag \\
    &=\lim_{N\to\infty} \!\!N\E \brac{\norm{\left (\Sigma_z^{1/2}F^T \!\otimes \!\Gamma \right)L\VEC \left(\bmat{\hat G_y \!-\! A\! & \!\!\hat G_u \!-\! B}  \right)}_2^2} 
\end{align}
where $L = \sum_{i = 1}^H e_i \otimes I_{\dx + \du} \otimes e_i \otimes I_{\dx}$ and $e_i$ is the $i$th column of $I_H$. From \Cref{lemma: ss single step var}, 
\begin{align} \label{eq: ss single step var}
    N \var \left ( \VEC \left(\bmat{\hat G_y \!-\! A\! & \!\!\hat G_u \!-\! B}\right) \right) \rightarrow   \Sigma_{x,u}^{-1} \otimes B_w B_w^\top 
\end{align}
where $\Sigma_{x,u}$ is the stationary covariance of $\bmat{x_t^\top & u_t^\top}^\top$. Expanding the norm and plugging in \eqref{eq: ss single step var}, \eqref{eq: proof II.2 loss} becomes
\begin{align*}
    \trace \left(\left(\Sigma_{x,u}^{-1} \otimes B_wB_w^\top \right)L^T 
    \left (F\Sigma_zF^T \otimes \Gamma^T\Gamma \right)L \right).
\end{align*}
Simplifying this expression gives the result stated in the proposition.
\end{proof}
Again, the error decays at a rate $1/N$. In contrast to the multi-step predictor, the asymptotic scaling of the single-step prediction error has the quantity $M_{SS} + \du I_H$ inside the trace. This means that the multi-step predictor suffers an extra factor of $H$ in the input term. Additionally the matrix $M_{MS}$ for the multi-step case has entries which decay as the distance to the diagonal increases, while $M_{SS}$ has entries which decay as the distance to the upper left element increases. Roughly, this  indicates that for very stable systems $M_{SS}$ should become smaller than $M_{MS}$. 

Next, we characterize the reducible error of the intermediate predictor. We introduce several quantities in order to cleanly express this characterization. Define the per-timestep loss
\begin{align}
 \label{eq: per timestep loss}
m_t(G_y,G_u)
&= \sum_{k=1}^H
\Bigl\|
y_{t+k}
- G_y^k y_t
- \sum_{i=0}^{k-1} G_y^{k-1-i} G_u\,u_{t+i}
\Bigr\|^2 .
\end{align}
Let
$
\label{eq: theta}
\theta := \VEC\!\left(\begin{bmatrix} G_y & G_u \end{bmatrix}\right).
$
Define
$
\label{eq: hessian}
J
= \bar{\E}\!\left[
\nabla_{\theta}^2\, m_t(A,B)
\right],
$
and
\begin{align}
\label{eq: jacobian}
\Sigma
&= \bar{\E}\!\left[
\sum_{\tau=-t}^{\infty}
\nabla_{\theta} m_t(A,B)\,
\bigl(\nabla_{\theta} m_{t+\tau}(A,B)\bigr)^{\!\top}
\right].
\end{align}
Let $\begin{bmatrix}\hat G_y^I & \hat G_u^I\end{bmatrix}$ be the first block row of the intermediate predictor $\hat G_N^I$. \Cref{lemma: ss_ms var} shows that the quantities $J$ and $\Sigma$ describe its asymptotic variance. In particular,
\begin{align} \label{eq: ss_ms var}
\lim_{N\to\infty}
N\,\var\!\left(\VEC\!\left(
\begin{bmatrix}
\hat G_y^I - A & \hat G_u^I - B
\end{bmatrix}
\right)\right)
= J^{-1}\Sigma J^{-1}.
\end{align}
The matrices $J$ and $\Sigma$ admit closed form expressions in terms of $A, B, B_w, $ and $H$. Implementations of these closed form expressions as required for the numerical validation in \Cref{s: well specified numerical} can be found in the accompanying codebase. We use these quantities in the following proposition. 
\begin{proposition}
    \label{prop: ss w/ ms loss well specified}
    The asymptotic error of the single-step predictor fitted with a multi-step loss $\hat G^{I}_N$ is given by 
    \begin{align*}
        &\lim_{N\to\infty} N \mathbf{E} \brac{\norm{(\hat G^{I}_N - G^\star) \Sigma_z^{1/2}}_F^2} \\
        &= \trace \left (J^{-1} \Sigma J^{-1}L^T 
    \left (F\Sigma_zF^T \otimes \Gamma^T\Gamma \right)L \right )
    \end{align*}
    where $L$, $F$, and $\Gamma$ are functions of $A, B, H$  as defined in Appendix~\ref{appendix: single step well specified}. 
\end{proposition} 

\begin{proof}
From the structure \eqref{eq: multistep loss}, we can write $\hat G_N^{I} - G^\star$ as a sum of terms which are linear in $\bmat{\hat G_y^I \!-\! A\! & \!\!\hat G_u^I \!-\! B}$ and higher order terms which vanish in the limit. That is,
\begin{align*}
    \hat G^{I}_N - G &= \Gamma (I_H \otimes (\bmat{\hat G_y^I - A & \hat G_u^I - B}  ))F + {\cal O}_p\paren{\frac{1}{N}}.
\end{align*}
 Thus, 
\begin{align*}
    &\lim_{N\to\infty} N \mathbf{E} \brac{\norm{(\hat G^{I}_N - G^\star) \Sigma_z^{1/2}}_F^2} \\
    &= \lim_{N\to\infty} N \mathbf{E} \brac{\norm{\Gamma \left(I_H \otimes \left(\bmat{\hat G_y^I - A & \hat G_u^I - B}  \right) \right)F \Sigma_z^{1/2}}_F^2} \\
    &= \lim_{N\to\infty} \!\!N\E \brac{\norm{\left (\Sigma_z^{1/2}F^T \!\otimes \!\Gamma \right)L\VEC \left(\bmat{\hat G_y^I \!-\! A\! & \!\!\hat G_u^I \!-\! B}  \right)}_2^2}. 
\end{align*} Expanding the norm and plugging in \eqref{eq: ss_ms var} gives the result. 
\end{proof}

Once again, the error decays asymptotically at a rate of $1/N$. This closed form expression does not admit an easy comparison with the single and multi-step decay rates. However, we are still able to show the comparison of the asymptotic variance of the three types of predictor.

\subsection{Comparison of Predictor Error}
\label{subsec: well-specified comparison}
\begin{figure*}
    \centering
    \includegraphics[width=0.32\linewidth]{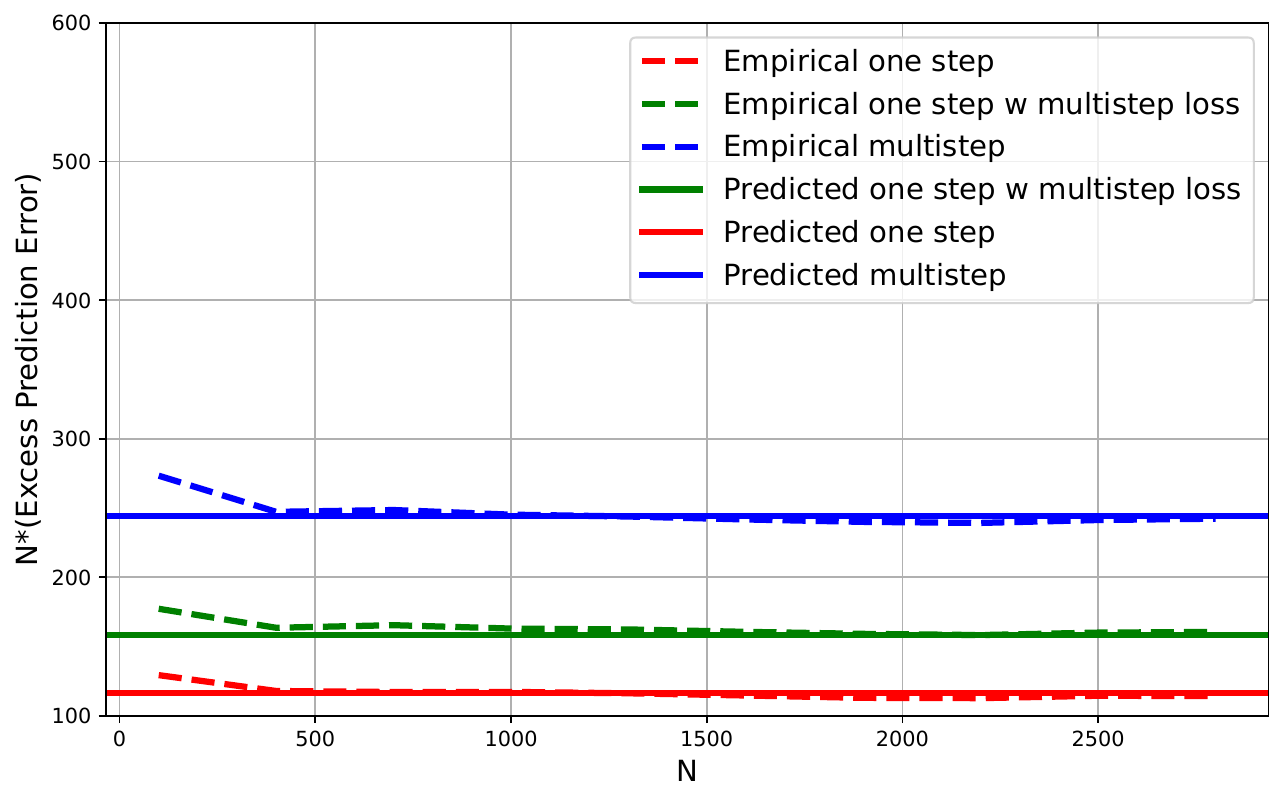}
    \includegraphics[width=0.32\linewidth]{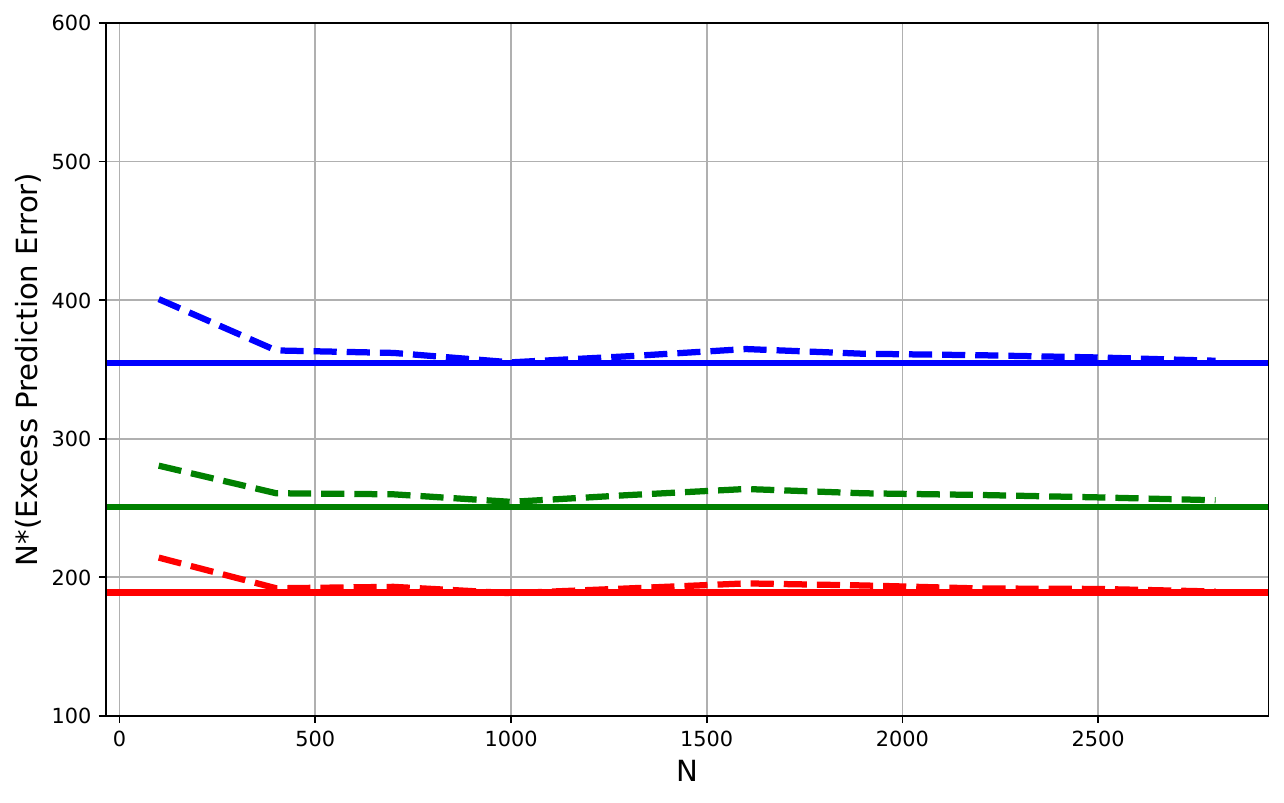}
    \includegraphics[width=0.32\linewidth]{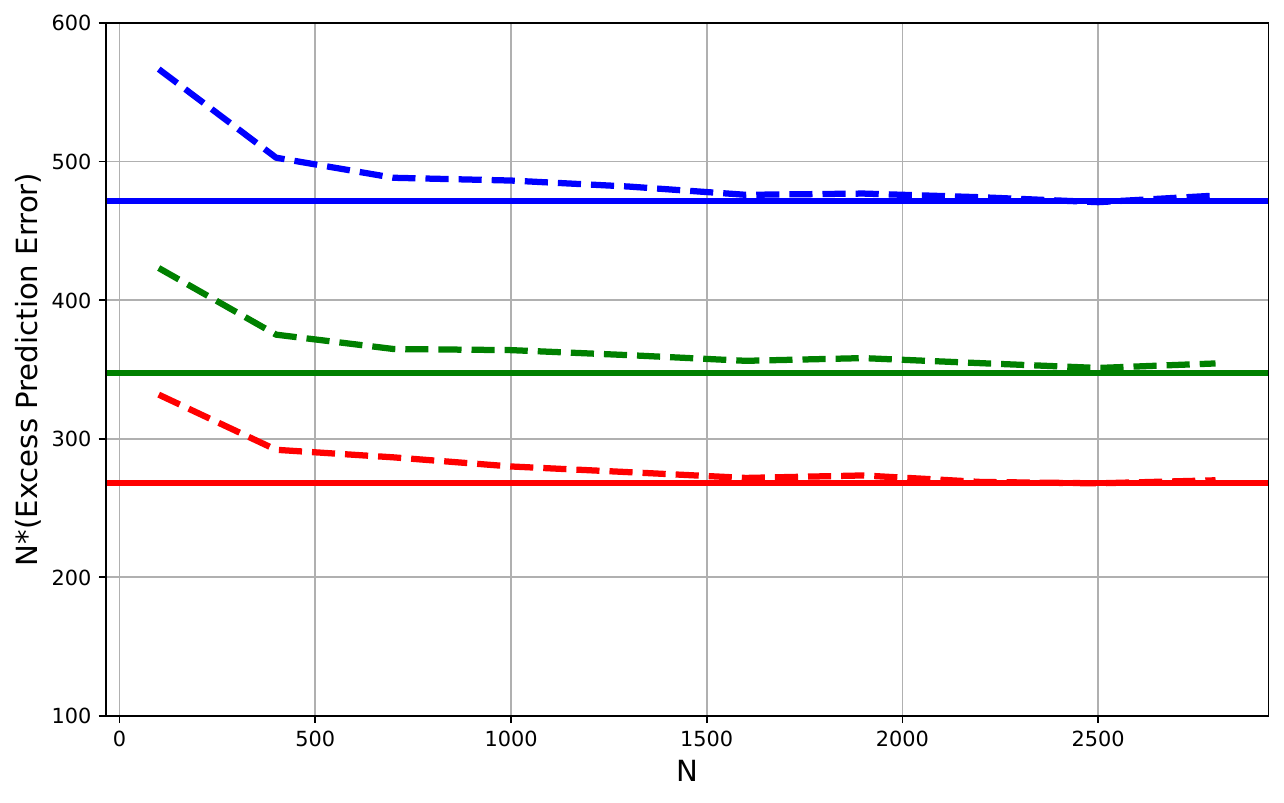}
    \caption{Convergence of $N\E[L(\hat f_H)]$ to the reducible prediction errors given in \Cref{prop: well specified multistep} (multi-step predictor), \Cref{prop: single-step well specified} (single-step rollout), and \Cref{prop: ss w/ ms loss well specified} (intermediate predictor) for the system defined by \Cref{eq: numerical ex system well specified} with $a = 0.5, 0.75,0.9$ (left to right) and horizon $H=5$.}
    \label{fig: well specified sweep a}
\end{figure*}
The next proposition compares the asymptotic decay rates given in Propositions~\ref{prop: well specified multistep}, \ref{prop: single-step well specified} and \ref{prop: ss w/ ms loss well specified}. 

\begin{proposition}
\label{prop: predictor ordering well specified}
The asymptotic decay rates for the single step, intermediate, and multistep predictors obey the ordering
    \begin{align*} 
        &\lim_{N\to\infty} N \mathbf{E} \brac{\norm{(\hat G^{SS}_N - G^\star) \Sigma_z^{1/2}}_F^2} \\
        &\leq \lim_{N\to\infty} N \mathbf{E} \brac{\norm{(\hat G^{I}_N - G^\star) \Sigma_z^{1/2}}_F^2} \\
        &\leq \lim_{N\to\infty} N \mathbf{E} \brac{\norm{(\hat G^{MS}_N - G^\star) \Sigma_z^{1/2}}_F^2}.
    \end{align*}
    
\end{proposition}

\begin{proof}
Consider the first inequality. Recall from the proofs of \Cref{prop: single-step well specified} and \Cref{prop: ss w/ ms loss well specified} that 
\begin{align*}
    &\lim_{N\to\infty} N \mathbf{E} \brac{\norm{(\hat G^{SS}_N - G^\star) \Sigma_z^{1/2}}_F^2} \\
    &= \trace ((\Sigma_{x,u}^{-1} \otimes B_wB_w^\top)L^\top(F\Sigma_zF^\top \otimes \Gamma^\top\Gamma)L)
\end{align*}
and 
\begin{align*}
    &\lim_{N\to\infty} N \mathbf{E} \brac{\norm{(\hat G^{I}_N - G^\star) \Sigma_z^{1/2}}_F^2} \\
    &= \trace \Bigl(J^{-1} \Sigma J^{-1} L^\top (F \Sigma_z F^\top \kron \Gamma^\top \Gamma) L\Bigl). 
\end{align*}
Notice that $\Sigma_{x,u}^{-1} \otimes B_wB_w^\top$ is the inverse of the Fisher Information in the data $\calD_N$ about the parameter $\VEC (\bmat{A & B})$. From the Cramér Rao inequality, $\Sigma_{x,u}^{-1} \otimes B_wB_w^\top \preceq J^{-1} \Sigma J^{-1}$. Thus, the first inequality of the proposition holds. 

Consider the second inequality. Let $\beta^\star = \VEC( G^\star)$. Let $\hat \beta^{MS} = \VEC(\hat G_N^{MS})$ be the parameters of the unconstrained multi-step predictor and $\hat \beta^{I} = \VEC(\hat G_N^{I})$ be the parameters of the intermediate predictor. Let $V_{\hat \beta^{MS}}$ be the asymptotic variance $\lim_{N \rightarrow \infty}N \var(\hat \beta^{MS} - \beta^\star) $ and similarly $V_{\hat \beta^{I}} \triangleq \lim_{N \rightarrow \infty}N \var(\hat \beta^{I} - \beta^\star)$. Let $Q \triangleq \Sigma_z \kron I_{H\dx}$. 
Then,  
\begin{align*}
    &\lim_{N\to\infty} N \mathbf{E} \brac{\norm{(\hat G^{I}_N - G^\star) \Sigma_z^{1/2}}_F^2} = \trace(V_{\hat \beta^{I}} Q)
\end{align*}
and 
\begin{align*}
    &\lim_{N\to\infty} N \mathbf{E} \brac{\norm{(\hat G^{MS}_N - G^\star) \Sigma_z^{1/2}}_F^2} = \trace(V_{\hat \beta^{MS}} Q). 
\end{align*}
 From results on constrained least squares regression (see \Cref{lemma: hansen} and \cite{Hansen2022Econometrics}), we have that 
\begin{align*}
    \trace(V_{\hat \beta^{I}} Q) = \trace(V_{\hat \beta^{MS}} Q) - \trace(V_{\hat \beta^{MS}}^{\frac{1}{2}}R(R^\top Q^{-1}R)^{-1}R^\top V_{\hat \beta^{MS}}^{\frac{1}{2}})
\end{align*}
for some $R \in \R^{H\dy(\dy + H\du) \times (H-1)\dy(\dy + H\du)}$ with full row rank.
Since $\trace(V_{\hat \beta^{MS}}^{\frac{1}{2}}R(R^\top Q^{-1}R)^{-1}R^\top V_{\hat \beta^{MS}}^{\frac{1}{2}}) \geq 0$, the result holds. 
\end{proof}
While Propositions~\ref{prop: well specified multistep}, \ref{prop: single-step well specified}, and \ref{prop: ss w/ ms loss well specified} explicitly characterize the asymptotic decay rates of the reducible errors induced by the three predictor classes, \Cref{prop: predictor ordering well specified} provides a direct comparison of these rates. In particular, it shows that for any system of the form \eqref{eq: dynamics} with full-state observations, the single-step predictor is statistically more efficient than the intermediate predictor (single-step model trained with a multi-step loss), which in turn is more efficient than the direct multi-step predictor.

This ordering reflects the extent to which each estimator exploits structural properties of the underlying dynamics. The single-step predictor benefits from (i) leveraging the Markovian structure of the true system and (ii) avoiding compounding process noise across multiple prediction steps in the training loss. The intermediate predictor also exploits the Markovian structure but incurs additional variance due to multi-step noise accumulation in the loss. The direct multi-step predictor benefits from neither advantage, and thus its prediction error exhibits the slowest asymptotic decay rate. 

In the case of the single-step and direct multi-step predictors, we can characterize and analyze the efficiency gap. Specifically, we can express the quadratic form defining the reducible portion of the error in \Cref{prop: well specified multistep} as the reducible error in \Cref{prop: single-step well specified} plus the additional term $\trace (\Gamma_w((M_{MS} - M_{SS} + (H-1)d_uI_H) \otimes I_{d_X})\Gamma_w^T)$. Note that $(M_{MS} - M_{SS})$ is the matrix with entry $(i,j)$ given by $\trace(\Sigma_x^{-1} \sum_{\ell=0}^{\min\curly{i,j}-2} A^{\ell}B_w B_w^\top (A^\ell)^\top (A^{\abs{j-i}})^\top)$. 
% \begin{align*}
%     \bmat{0  & 0 & 0 & \dots \\
%     0 & \trace\paren{\Sigma_x^{-1} B_wB_w^\top } &  \trace\paren{\Sigma_x^{-1} B_wB_w^\top A^\top} & \dots \\
%     0 & \trace\paren{\Sigma_x^{-1} B_wB_w^\top A^\top} &   \sum_{\ell=0}^1 \trace\paren{\Sigma_x^{-1} A^\ell B_wB_w^\top (A^\ell)^\top} & \dots \\ \vdots && \ddots }.
% \end{align*}
We see that
%in the more general setting of a vector valued state and inputs, 
the efficiency gap grows with the prediction horizon $H$. This scaling quantitatively captures that the direct multi-step predictor has a number of parameters which scales with $H$. 

To better understand the role of system stability in determining this efficiency gap, we consider the special case of a scalar system without inputs.  Here the difference between the statistical efficiency of the single-step predictor and the multi-step predictor is characterized by the difference between the matrices
\begin{align*}
    M_{SS} =  \bmat{1 & a & a^2 & \dots &a^{H-1} \\ a & a^2 & a^3 &\dots &a^{H} \\ \vdots & & \ddots \\ 
    a^{H-1} & &\dots && a^{2(H-1)}}
\end{align*}
and
\begin{align*}
    M_{MS} = \bmat{1 & a & a^2 & \dots &a^{H-1} \\ a & 1 & a &\dots &a^{H-2} \\ \vdots & & \ddots \\ 
    a^{H-1} & &\dots && 1},
\end{align*}
from which we conclude that the difference between the two diminishes as $\abs{a} \to 1$. Since the statistical efficency of the intermediate predictor lies in between that of the single step and multistep predictors by \Cref{prop: predictor ordering well specified}, we can conclude that the difference between all three predictor classes diminishes as the system approaches marginal stability.

The above statement about the decay in the efficiency gap as the system approaches marginal stability holds for scalar systems, but the picture is more nuanced in general. In particular, if any eigenvalue remains less than 1, we maintain an efficiency gap even as $\rho(A) \rightarrow 1$, as seen in \Cref{fig: well specified sweep a}.

\subsection{Numerical Experiments}
\label{s: well specified numerical}

To validate the error characterizations presented in \cref{sec: well-specified}, we consider the fully observed system defined by 
\begin{align} \label{eq: numerical ex system well specified}
    A = \bmat{a & 1.0 \\ 0.0 & 0.75}, \Sigma_w = I_2
\end{align}
with $B = \bmat{0 & 1}^\top, C = I_2, \Sigma_v = 0$.

In \Cref{fig: well specified sweep a}, we estimate $N\E[L(\hat f_H)]$ by averaging over $2,500$ datasets $D_N$ for $N \in \{ 1,...,3000\}$ to demonstrate convergence to the reducible errors given in Propositions~\ref{prop: well specified multistep}, \ref{prop: single-step well specified} and \ref{prop: ss w/ ms loss well specified}. In these figures, we fix $H=5$ and vary $a$ across $0.5, 0.75$, and $0.9$. For the single-step model with a multi-step loss, $\hat G$ is fit using Adam initialized from the single-step predictor fit with a single-step loss, and using a step size $1e-2$. 

For each data generating process (choice of parameter $a$) and each predictor, we see that the decay rate of the excess prediction error converges to its predicted value as $N$ increases. Comparing the decay rates across predictor classes, these results exemplify the efficiency gap suggested by \Cref{prop: predictor ordering well specified}. That is, for each choice of the parameter $a$, the excess prediction error of the single-step predictor has the fastest decay rate, followed by the single-step predictor with a multi-step loss, then the direct multi-step predictor. Comparing the figures left to right, we see that, unlike for the scalar system discussed in \Cref{subsec: well-specified comparison}, the efficiency gap does not diminish with $a$.

\section{Misspecified Setting}
\label{sec: misspecified}

We again consider a single-step predictor and a multi-step predictor applied to the measurement. However, we now examine the general case in which measurements do not provide full state information by reincorporating partial observations, as specified by $C \in \R^{\dy \times \dx}$ and $D_v D_v^\top \succ 0$, into the dynamics~\eqref{eq: dynamics}. Due to the Markovian assumption made in fitting the predictor, this represents a misspecified setting. To ease notational burden we restrict attention to the setting without inputs and set $B=0$, although our analysis can be extended naturally to the $B \neq 0$ case.

In this setting, the loss is given by
\begin{align}
    \label{eq: misspecified prediction error}
    L(\hat G) = \bar \E \norm{y_{t+1:t+H} - \hat G y_t }^2.
\end{align}
We rewrite the dynamics in innovations form as
\begin{align*}
    \hat x_{t+1} &= (A-KC) \hat x_t + K y_t = A \hat x_t + K D_e e_t \\
    y_t &= C \hat x_t + D_e e_t, 
\end{align*}
where $e_t$ is standard normal noise that is independent across time, $K$ is the Kalman gain defined as $K = ASC^\top (C S C^\top + R)^{-1}$, $S$ is the stabilizing solution to the Riccati equation defined by $A$, $C$, $D_w D_w^\top$ and $D_v D_v^\top$, and $D_e = (C S C^\top + D_v D_v^\top)^{1/2}$.  Then
\begin{align*}
    y_{t+1:t+H} = \Phi \hat x_t + G^\star y_t + \Gamma_e e_{t+1:t+H},
\end{align*}
where 
\begin{align*}
    \Phi &= \bmat{C (A-KC) \\ C A (A-KC) \\ \vdots \\ CA^{H-1} (A-KC)}, G^\star = \bmat{CK \\ C A K \\ \vdots \\ CA^{H-1} K}, \\
    \Gamma_e &= \bmat{D_e & 0 & \dots & 0\\ C K D_e & D_e & \dots \\ \vdots && \ddots \\ CA^{H-2} K D_e & \dots & CK D_e & D_e}. 
\end{align*}
Under these definitions, and exploiting that innovations are independent across time, the error~\eqref{eq: misspecified prediction error} is given by 
\begin{align*}
    L(\hat G) &= \bar \E\norm{\Phi \hat x_t + (G^\star - \hat G) y_t + \Gamma_e e_{t+1:t+H}}^2 \\
    &= \bar \E\norm{\Phi \hat x_t + (G^\star - \hat G) y_t}^2 + \norm{\Gamma_e}_F^2.
\end{align*}
% by the fact that the innovations are independent of the past. 
Expanding $y_t = C \hat x_t + D_e e_t$,

\begin{align*}
    &L(\hat G) \\
    &=\! \norm{(\Phi  \!+\! ( G^\star - \hat G)C)\Sigma_{\hat x}^{1/2}}_F^2 + \norm{(G^\star - \hat G) D_e}_F^2  \!+\! \norm{\Gamma_e}_F^2, 
\end{align*}
where $\Sigma_{\hat x}$ is the stationary covariance of $\hat x_t$. Let $\Sigma_y$ be the stationary covariance of $y_t$.\footnote{It follows from $D_v D_v^\top \succ 0$ that $\Sigma_y$ is positive definite. }
When $\hat G$, or equivalently $\hat G_N$, is learned on the dataset of size $N$, we can decompose this quantity into an irreducible component, and a component which decays to zero as the amount of data $N\to \infty$. Denoting the irreducible component by $B(\hat G_N) \triangleq \lim_{N\to\infty} \E L(\hat G_N)$ and the reducible component by $\varepsilon(\hat G_N) \triangleq L(\hat G_N) - B(\hat G_N)$, we decompose
\begin{align*}
    L(\hat G_N) = B(\hat G_N) + \varepsilon(\hat G_N). 
\end{align*}
Unlike the well-specified setting, the irreducible component $B(\hat G_N)$ differs depending on whether we fit a single-step model, direct multi-step model, or single-step model fitted with a multi-step loss. We therefore focus on comparing these bias terms rather than the rate of convergence, since this captures the fundamental difference between the two models. 
% See Appendix~\ref{appendix: misspecified} for characterization of the rate of decay of the reducible error $\lim_{N \rightarrow \infty} N\E[\varepsilon(\hat G_N) ]$ for the direct multi-step predictor. 

% \begin{figure}
%     \centering
%     \includegraphics[width=\linewidth]{figures/a_90_H_10.pdf}
%     \caption{Convergence of $N\e_{D_N}[L(\hat f_H)]$ to the reducible prediction errors given in \Cref{prop: well specified multistep} (multi-step predictor) and \Cref{prop: single-step well specified} (single-step rollout) for the system defined by \Cref{eq: numerical ex system} with $a = 0.9$ (left to right) and horizon $H=10$.}
%     \label{fig: well specified H 10}
% \end{figure}

\begin{figure*}
    \centering
    \includegraphics[width=0.32\linewidth]{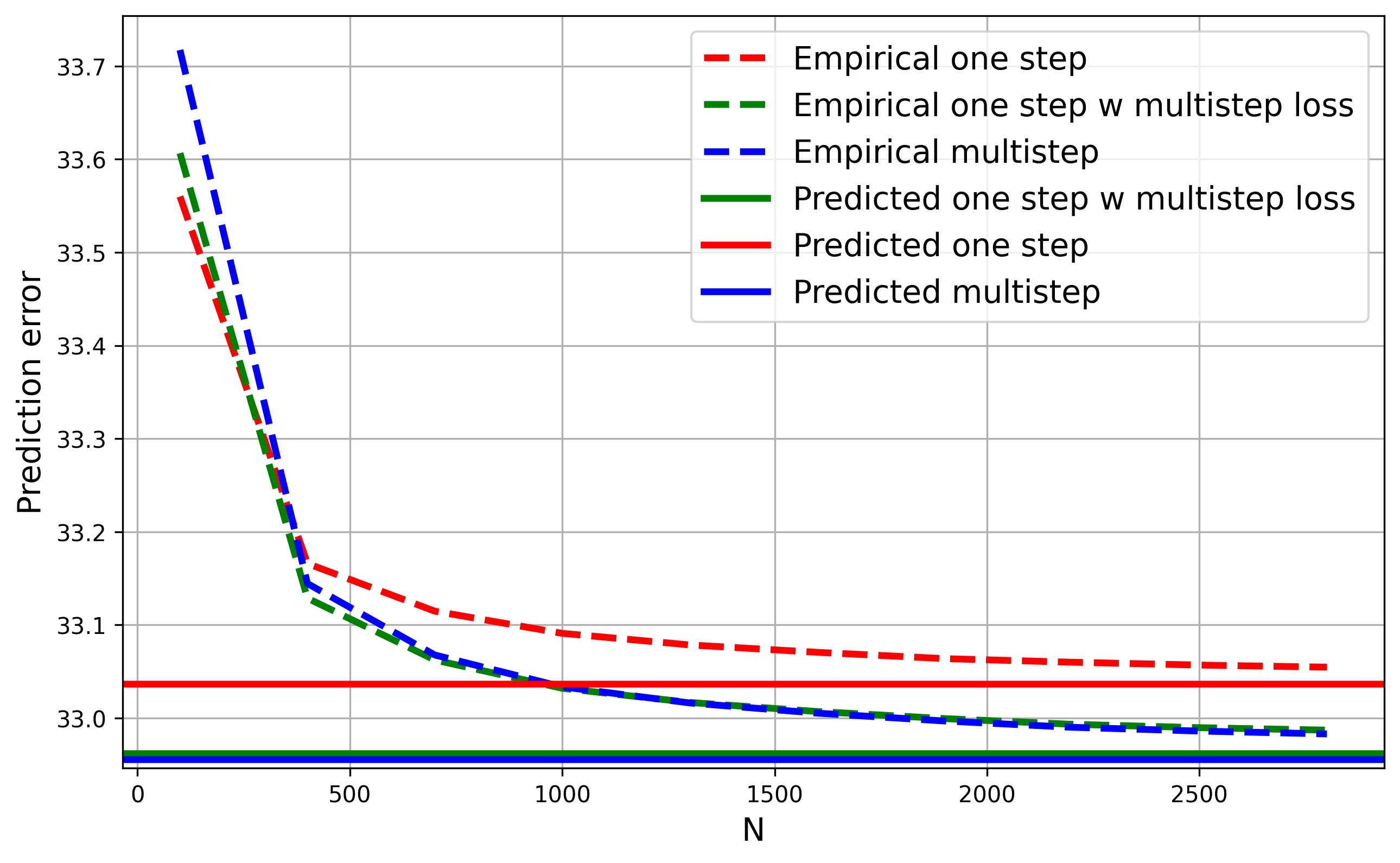}
    \includegraphics[width=0.32\linewidth]{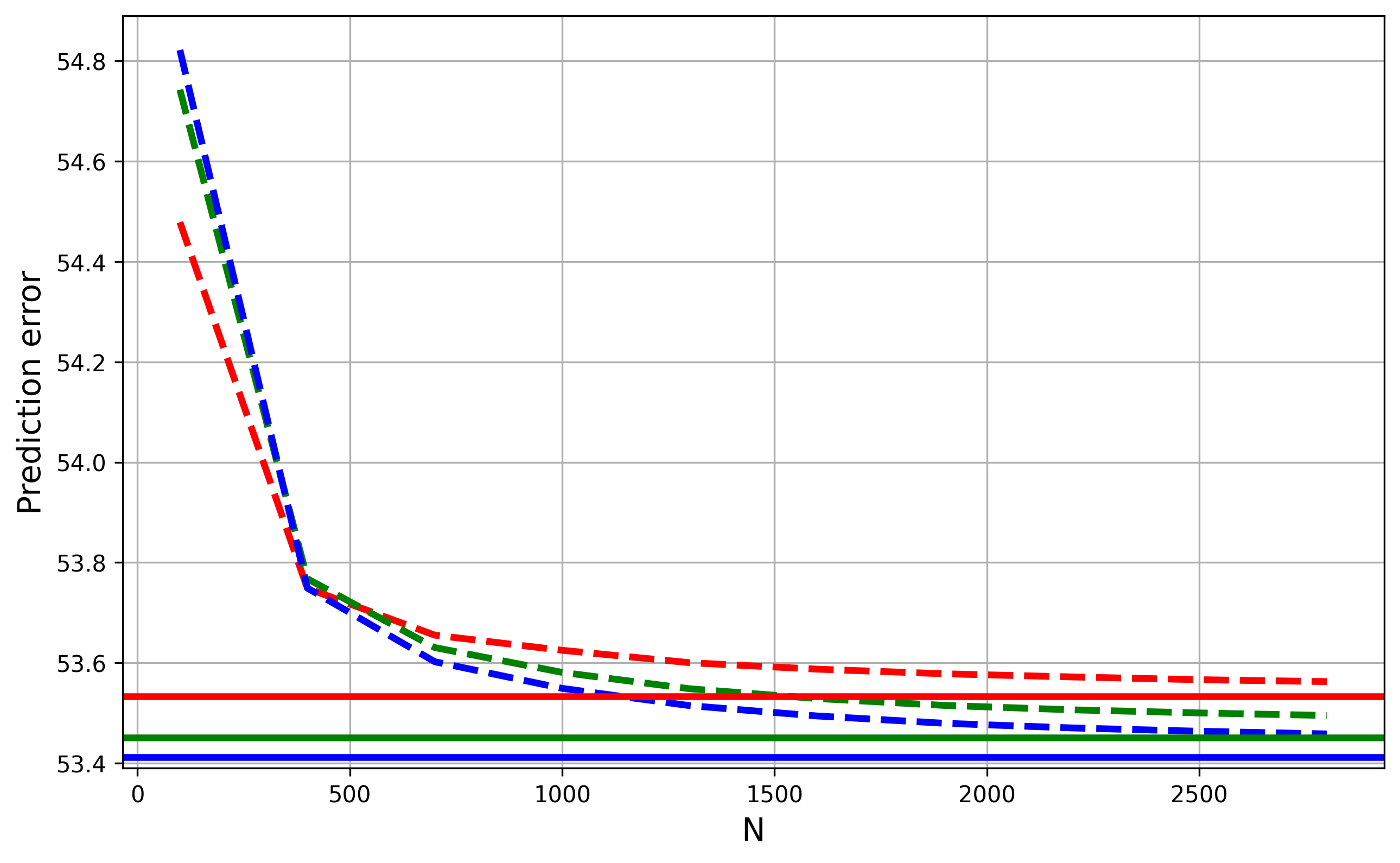}
    \includegraphics[width=0.32\linewidth]{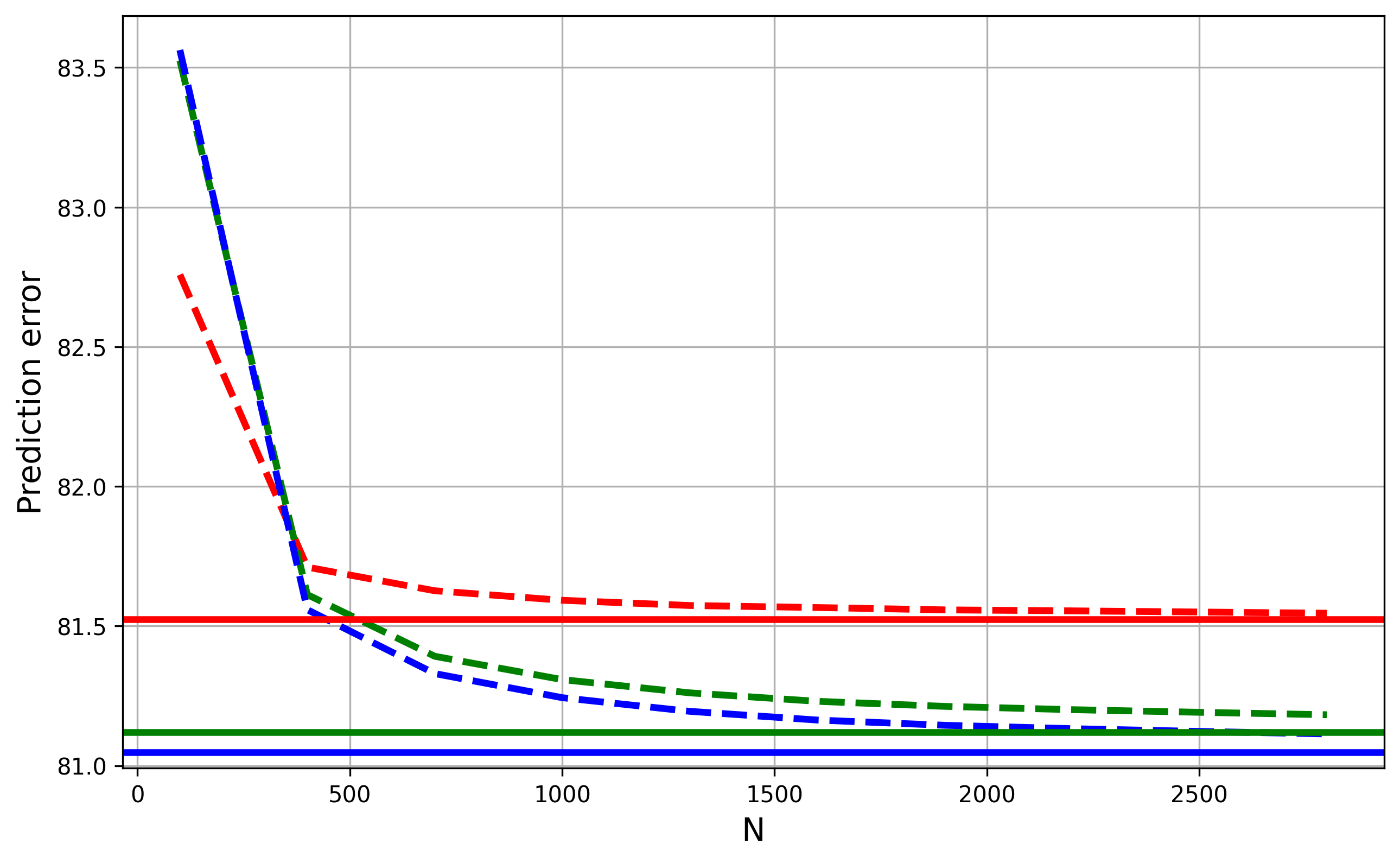}
      \caption{Convergence of $\E[L(\hat f_H)]$ to the irreducible prediction errors given in \Cref{prop: multi-step misspecified} (multi-step predictor),  \Cref{prop: single-step misspecified} (single-step rollout), and \Cref{prop: single-step w/ multi-step loss misspecified} (intermediate predictor) for the system defined by \Cref{eq: numerical ex system misspecified} with $a = 0.5, 0.75,0.9$ (left to right) and horizon $H=5$.}
      \vspace{-16pt}
    \label{fig: misspecified sweep a}
\end{figure*}

The irreducible error for the multi-step predictor is characterized in the following proposition.
\begin{proposition}[Proposition~IV.1 of \cite{CDCPaper}]
    \label{prop: multi-step misspecified}
    The irreducible error for the multi-step predictor $\hat G_N^{MS}$ is given by 
    \begin{align*}
         % \lim_{N\rightarrow \infty} \E[B(\hat G_N^{MS})] 
         B(\hat G_N^{MS}) \!=\! \trace(\Phi (\Sigma_{\hat x} - \Sigma_{\hat x} C^\top \Sigma_y^{-1} C \Sigma_{\hat x}) \Phi^\top) \! +\! \norm{\Gamma_e}_F^2, 
    %\!+\! \varepsilon_N, 
    % \end{align*}
    % where 
    % \begin{align*}
    %     &\lim_{N\to \infty} \!\!\! N \!\e_{D_N}\! \brac{\varepsilon_N } \\
    %     &=\! \trace(\Phi \Sigma_{\hat X} \Phi^\top\!)\! \trace(D_e^\top \Sigma_{Y}^{-1} D_e) \!+\! \trace\paren{\Gamma_e (M_1 \otimes I_{\dy}) \Gamma_e^\top\!} \\ 
    %     &+2 \trace\Bigg(\Phi M_2 (I_H \otimes (A \Sigma_{\hat X} C^\top + K D_e D_e^\top) \Sigma_Y^{-1} D_e )\Gamma_e^\top\Bigg),
    % \end{align*}
    % \begin{align*}
    %     M_1 &\!=\!\!\bmat{\trace(I) & \trace(\Sigma_Y^{(1)}) & \dots & &\trace(\Sigma_Y^{(H-1)}) \\ 
    %     \trace(\Sigma_Y^{(1)}) & \trace(I) & \trace(\Sigma_Y^{(1)}) & \!\!\dots\!\! & \!\!\trace(\Sigma_Y^{(H-2)}) \!\!\\ \vdots & & &\!\!\!\ddots\!\!\! \\ \trace(\Sigma_Y^{(H-1)})\!\!\! & \dots & & &\trace(I)
    %     }, \\
    %     \Sigma_Y^{(i)} &= (C A^{i} \Sigma_{\hat X} C^\top + CA^{i-1} K D_e D_e^\top) \Sigma_Y^{-1}, \\
    %     M_2 &= \bmat{I &A & \dots & A^{H-1}}. 
    \end{align*}
    where $\Sigma_{\hat x}$ is the stationary covariance of $\hat x_t$.
\end{proposition}

\begin{proof}
    
The least squares identification error is 
% \begin{align*}
%     \hat G = \sum_{t=1}^{N-H} y_{t+1:t+H} y_t^\top \paren{\sum_{t=1}^{N-H} y_t y_t^\top}^{-1},
% \end{align*}
% and thus
\begin{align*}
    \hat G_N^{MS}\! -\! G^\star \!=\! \sum_{t=1}^{N-H} (\Phi \hat x_t + \Gamma_e e_{t+1:t+H}) y_t^\top \left(\sum_{t=1}^{N-H} y_t y_t^\top\right)^{-1}.
\end{align*}
%\Bruce{These arguments need to be made concrete, but take it for now}
% By the law of large numbers paired with Slutsky's theorem, we can substitute appearances of this gap appearing in the asymptotic error characterization with
% \begin{align*}
%     \hat G_N - G^\star \approx \frac{1}{N-H}\sum_{t=1}^{N-H} (\Phi \hat x_t + \Gamma_e e_{t+1:t+H}) y_t^\top \Sigma_y^{-1}
% \end{align*}
Expanding $y_t = C \hat x_t + D_e e_t$, this becomes
\begin{align} \label{eq: ems breakdown}
    &\Phi \left (\sum_{t=1}^{N-H}\hat x_t \hat x_t^\top \right )C^\top\left(\sum_{t=1}^{N-H} y_t y_t^\top\right)^{-1} + \tilde E_{MS},
\end{align}
where 
\begin{align} \label{eq: ems}
    &\tilde E_{MS} \notag\\
    &\triangleq \left (\sum_{t=1}^{N-H} \Phi \hat x_t e_t^\top D_e^\top  \!+\!   \Gamma_e e_{t+1:t+H}  y_t^\top \right )\left(\sum_{t=1}^{N-H} y_t y_t^\top\right)^{-1}.
\end{align}
By Slutsky's theorem and the Birkhoff Khinchin Theorem, 
\begin{align} \label{eq: multistep ls error bias}
    & \Phi \left (\sum_{t=1}^{N-H}\hat x_t \hat x_t^\top \right )C^\top\left(\sum_{t=1}^{N-H} y_t y_t^\top\right)^{-1} \probconv \Phi \Sigma_{\hat x} C^\top \Sigma_y^{-1}.
\end{align}
Plugging this in to $L(\hat G_N^{MS})$, applying Vitali Convergence, and noting that terms which scale with $\tilde E_{MS}$ vanish in the limit, it follows that 
\begin{align*}
&\lim_{N\to\infty} \E L(\hat G_N^{MS}) \\ &= \trace(\Phi (\Sigma_{\hat x} - \Sigma_{\hat x} C^\top \Sigma_y^{-1} C \Sigma_{\hat x}) \Phi^\top) + \norm{\Gamma_e}_F^2.
\end{align*}
\end{proof}
The above proposition quantifies exactly how the loss of information from partial system observations (through $C \neq I$ or $D_v \neq 0$) affects the asymptotic bias induced by the multi-step predictor. It shows that this bias scales with the prediction horizon $H$ as $A^H$, suggesting worse prediction performance for systems with larger spectral radius $\rho(A)$. 

We can compare this with the irreducible error associated with the single-step predictor, which is characterized as follows.
\begin{proposition}[Proposition~IV.2 of \cite{CDCPaper}]
    \label{prop: single-step misspecified}
    The irreducible error for the single-step predictor $\hat G_N^{SS}$ is given by  
    % \begin{align*}
    %     % \lim_{N \rightarrow \infty} \E [B(\hat G_N^{SS})]
    %     B(\hat G_N^{SS}) &= \trace((\Phi + (G^\star - G_\infty)C) \Sigma_{\hat x} (\Phi + (G^\star - G_\infty)C)^\top) \\
    %     &+ \trace((G^\star - G_\infty) D_e D_e^\top (G^\star - G_\infty)^\top) + \norm{\Gamma_e}_F^2, 
    % \end{align*}
    \begin{align*}
        % \lim_{N \rightarrow \infty} \E [B(\hat G_N^{SS})]
        B(\hat G_N^{SS}) &= \trace((\Phi + (G^\star - \tilde G)C) \Sigma_{\hat x} (\Phi + (G^\star - \tilde G)C)^\top) \\
        &+ \trace((G^\star - \tilde G) D_e D_e^\top (G^\star - \tilde G)^\top) + \norm{\Gamma_e}_F^2, 
    \end{align*}
    where 
    \begin{align}
        \label{eq: multi-step bias characterization}
        \tilde{G} = \bmat{ C A \Sigma_{\hat x} C^\top \Sigma_y^{-1} \\ \vdots \\ (C A \Sigma_{\hat x} C^\top \Sigma_y^{-1})^{H} },
    \end{align}
    and $\Sigma_x$ is the stationary covariance of $x_t$. 
\end{proposition}
\begin{proof}
   Since we consider the setting without inputs, we can write 
   \begin{align*}
       \hat{G}_N^{SS} &= \bmat{\hat G_y \\ \hat G_y^2 \\ \vdots \\ \hat G_y^H}
   \end{align*} where, from the normal equations for the least squares estimator,  
\begin{align*}
    &\hat G_y = \sum_{t=1}^{N-1} y_{t+1} y_t^\top \paren{\sum_{t=1}^{N-1} y_t y_t^\top}^{-1}.
\end{align*}
Expanding $y_{t+1}$ and $y_t$, this becomes 
\begin{align*}
    &CK \!+\!   C(A\!-\!KC) \paren{\sum_{t=1}^{N-1}\hat x_t \hat x_t^\top} C^\top \paren{\sum_{t=1}^{N-1} y_t y_t^\top}^{-1} \!+ \!\tilde E_{SS},
\end{align*}
where 
\begin{align*}
    &\tilde E_{SS} \!\triangleq \sum_{t=1}^{N-1} \big ( \left (C(A-KC \right) \hat x_t e_t^\top D_e^\top \\
    &+ D_e e_{t+1} y_t^\top + D_e e_{t+1}e_t^\top D_e^\top \big ) \paren{\sum_{t=1}^{N-1} y_t y_t^\top}^{-1}. 
\end{align*}
By Slutsky's theorem and the Birkhoff Khinchin theorem, 
\begin{align*}
    &CK +   C(A-KC) \paren{\sum_{t=1}^{N-1}\hat x_t \hat x_t^\top} C^\top \paren{\sum_{t=1}^{N-1} y_t y_t^\top}^{-1} \\
    & \probconv C A \Sigma_{\hat x} C^\top \Sigma_y^{-1}.
\end{align*}

Plugging this in to $L(\hat G_N^{SS})$, applying Vitali Convergence, and noting that terms which scale with $\tilde E_{SS}$ vanish in the limit, it follows that 
\begin{align*}
&\lim_{N\to\infty} \E L(\hat G_N^{SS}) \\
&= \trace \left((\Phi + (G^\star - \tilde G)C) \Sigma_{\hat x} (\Phi + (G^\star - \tilde G)C)^\top \right) \\
        &+ \trace \left ((G^\star - \tilde G) D_e D_e^\top (G^\star - \tilde G)^\top \right) + \norm{\Gamma_e}_F^2.
\end{align*}

\end{proof}

The above proposition quantifies how the loss of information from partial system observations affects the asymptotic bias induced by the single-step predictor. It shows that this bias scales with the prediction horizon $H$ as $(C A \Sigma_{\hat x} C^\top \Sigma_y^{-1})^{H}$. As in the case of the direct multi-step predictor, this suggests worse prediction performance for systems with larger spectral radius $\rho(A)$. However, unlike the case of the multi-step predictor, this scaling depends on $C, \Sigma_y,$ and $\Sigma_{\hat x}$, meaning that the quality of system observations affects how the bias scales with the horizon. See \Cref{example} for a numerical example where this dependence leads to prohibitive compounding error.

In the case of the intermediate predictor, the irreducible error can be written as a solution of a constrained optimization problem.
\begin{proposition}
    \label{prop: single-step w/ multi-step loss misspecified}
    Let $\hat G_N^{I}$ be the single-step predictor with a multistep loss defined as in \eqref{eq: multistep LS} with the additional constraint that $\lVert G_N^I \rVert_F < R$ for some $R \in \R$ for all $N$. The irreducible error solves the optimization problem  
    % \begin{align*}
    %     B(\hat G_N^{I}) &= \min_{G \in S}\trace((\Phi + (G^\star - G)C) \Sigma_{\hat x} (\Phi + (G^\star - G)C)^\top) \\
    %     &+ \trace((G^\star - G) D_e D_e^\top (G^\star - G)^\top) + \norm{\Gamma_e}_F^2
    % \end{align*}
    \begin{align*}
        B(\hat G_N^{I}) &= \min_{G \in S}L(G)
    \end{align*}
    where 
   $S$ is defined as in \eqref{eq: multistep loss} and we assume $R$ is chosen large enough that $\argmin _{G \in S} L(G) \subset \{G \mid \lVert G \rVert_F < R \} $. 
\end{proposition}
Note that there is no closed form solution to the optimization problem. Also note that $R$ may be chosen arbitrarily large, so the restriction to the ball $\{G \mid \lVert G \rVert_F < R \}$ is practically inconsequential. 
\label{proof: single-step w/ multi-step loss misspecified}
\begin{proof}
    Define 
    \begin{align*}
        L_N(G) = \frac{1}{N-H-1}\sum_{t=1}^{N-H} \norm{y_{t+1:t+H} - G y_t}^2
    \end{align*}
    so that 
    \begin{align*}
        \hat G_N^I \in \argmin_{G \in S, \lVert G \rVert_F <R} L_N(G). 
    \end{align*}
    Let 
    \begin{align*}
        L^* = \min_{G \in S} L(G). 
    \end{align*}
    From Lemma 2.4 of \cite{NeweyMcFadden1999}, 
    \begin{align*}
        \sup_{G \in S, \lVert G \rVert_F <R} \abs{L_N(G) - L(G)} \probconv 0.
    \end{align*}
     For $G^* \in \argmin_{G \in S} L(G),$
    \begin{align*}
    L(\hat G_N^I) - L^*
    &= \bigl(L(\hat G_N^I)-L_N(\hat G_N^I)\bigr) \\
    & + \bigl(L_N(\hat G_N^I)-L_N(G^*)\bigr)
       + \bigl(L_N(G^*)-L^\star \bigr) \\
    &\le 2 \sup_{\|G\|<R} |L(G)-L_N(G)|
    \;\probconv\; 0 .
    \end{align*}
    Thus, 
    \begin{align*}
        L(\hat G_N^I) \probconv L^*. 
    \end{align*}
    The final result holds by uniform integrability of $ L(\hat G_N^I)$ and Vitali's Convergence Theorem. 
\end{proof}

Though it is impossible to give a closed form solution to the optimization problem in the above problem, this characterization of the bias is useful in the next section where we compare the biases induced by the three classes of predictors. 

\subsection{Comparison}
The next proposition compares the asymptotic biases given in \Cref{prop: multi-step misspecified}, \Cref{prop: single-step misspecified}, and \Cref{prop: single-step w/ multi-step loss misspecified}.
\begin{proposition}
\label{prop: predictor ordering misspecified}
It holds that
    \begin{align*}
         B(\hat G_N^{MS}) \leq  B(\hat G_N^{I}) \leq  B(\hat G_N^{SS}).
    \end{align*}
    
\end{proposition}
This ordering follows from the observation that the irreducible errors of the multi-step, intermediate, and single-step predictors can be written as solutions of nested optimization problems. Specifically, the bias from the multi-step predictor given in \Cref{prop: multi-step misspecified} is given by 

\begin{align} \label{eq: multistep optimization prob}
    B(\hat G_N^{MS}) = \min_{G \in \R^{H\dy \times \dy}} L(G). 
\end{align}
The bias from the intermediate predictor given in \Cref{prop: single-step w/ multi-step loss misspecified} solves the same problem but with $G$ constrained to the set $S$ defined in \eqref{eq: multistep loss}. 
The bias from the single step predictor in \Cref{prop: single-step misspecified} corresponds to an evaluation of the same loss at a fixed point $\tilde G \in S$. That is,   $B(\hat G_N^{SS}) =L(\tilde{G})$. Since, $\R^{H \dy \times \dy} \supseteq S \supseteq \{\tilde G\}$, the ordering holds. 

In the case of the single-step predictor and direct multi-step predictor, we can quantify how the bias scales with the horizon $H$. As discussed earlier, the single step predictor scales with the horizon as $(C A \Sigma_{\hat x} C^\top \Sigma_y^{-1})^{H}$ and the direct multi-step predictor as $A^H$. Thus, the scaling is dictated by the spectral radii of these quantities. 
\Cref{lemma: spectral radius of CK+KB} shows that the quantity $C A \Sigma_{\hat x} C^\top \Sigma_y^{-1}$ which the estimate $\hat G_y$ converges to satisfies $\rho(C A \Sigma_{\hat x} C^\top \Sigma_y^{-1}) \leq 1$ if $\rho(A)<1$.  Despite this, $C A \Sigma_{\hat x} C^\top \Sigma_y^{-1}$ can feature a spectral radius much larger than that of $A$, as demonstrated in the following example.  
\begin{example}
    \label{example}
    \begin{figure}
    \centering
    \includegraphics[width=\linewidth]{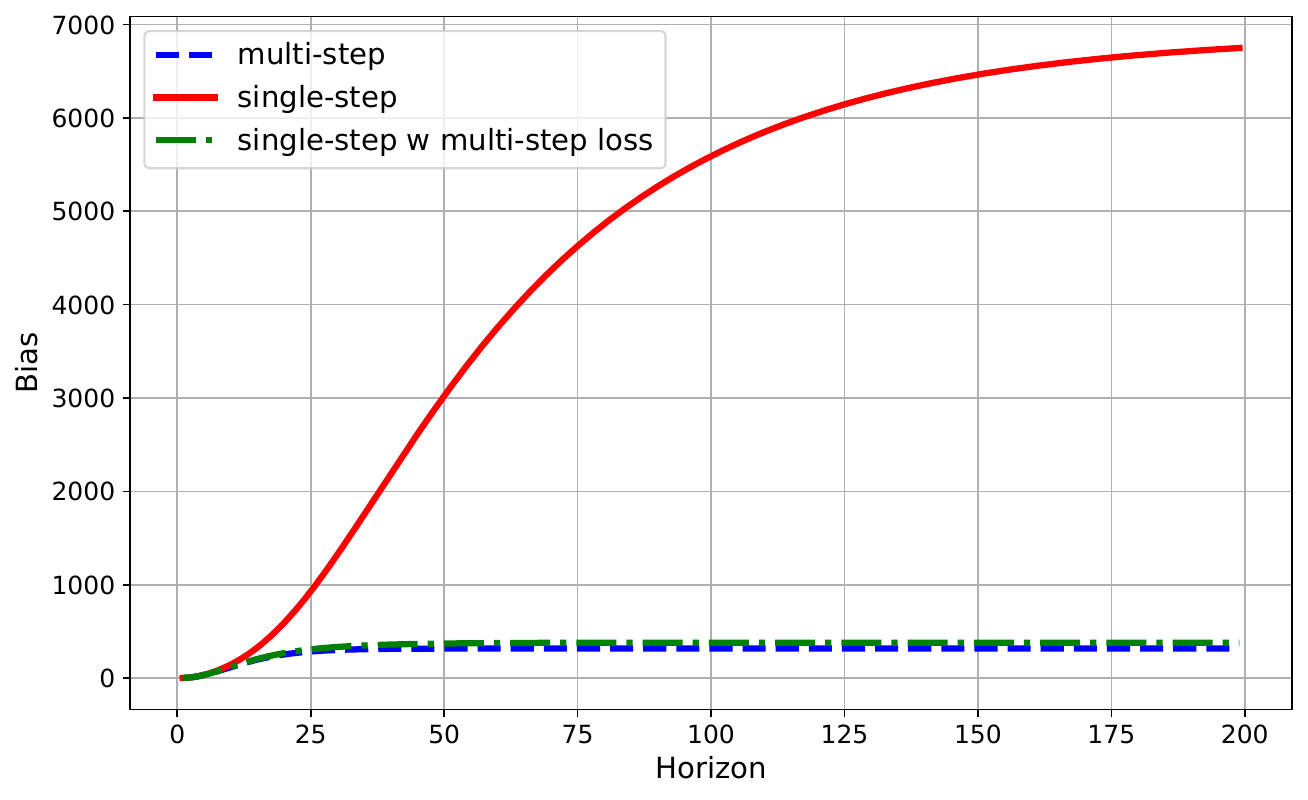}
    \caption{Comparison of the bias from the single and multi-step estimators across different horizons for the system defined in Example~\ref{example}.}
    \label{fig: bias comparison}
\end{figure}
    Consider the system defined by
    \begin{align*}
        A = \bmat{0.9 & 1.0 \\ 0.0 & 0.9}, \Sigma_w = I_2, C = \bmat{1, 0}, \Sigma_v = 1.
    \end{align*}
    We find that $\rho(C A \Sigma_{\hat x} C^\top \Sigma_y^{-1}) = 0.99$, though $\rho(A) = 0.9$. 
\end{example}
As a consequence of this fact, the gap in bias between the multi-step error and the single-step error can grow with the horizon for moderate $H$. This is illustrated for the above example in \Cref{fig: bias comparison}.

\subsection{Numerical Experiments}
To validate the error characterizations presented in \cref{sec: misspecified}, we consider the partially observed system defined by 
\begin{align} \label{eq: numerical ex system misspecified}
    A = \bmat{a & 1.0 \\ 0.0 & 0.75}, \Sigma_w = I_2
\end{align}
with $ B = 0, C = [1,0], \Sigma_v = 1$. 

In \Cref{fig: misspecified sweep a}, we estimate $\E[L(\hat f_H)]$ by averaging over $2,500$ datasets $D_N$ for $N \in \{ 1,...,3000\}$ to demonstrate convergence to the irreducible errors given in \Cref{prop: multi-step misspecified} and \Cref{prop: single-step misspecified}. In these figures, we fix $H=5$ and vary $a$ across $0.5, 0.75$, and $0.9$. For the single-step model with a multi-step loss, $\hat G$ is fit using Adam initialized from the single-step predictor fit with a single-step loss, and using a step size $1e-2$.  

For each data generating process (choice of parameter $a$) and each predictor, we see that the prediction error converges to its predicted value as $N$ increases. Comparing prediction error across predictor classes, these results exemplify the bias gap suggested in \Cref{prop: predictor ordering misspecified}. That is, for each choice of the parameter $a$, the multi-step predictor has the lowest asymptotic prediction error, followed by the single-step predictor with a multi-step loss, then the single-step predictor.

% To see this, note that in the case of scalar measurements, the gap in the bias between the single-step predictor and the multi-step predictor is lower bounded by 
% \begin{align*}
%     &  \sum_{h=\ell}^H \Bigg(\Sigma_Y \abs{CA^{h-1}K - (CK + K_B)^{h}}- \\
%     &2\! \norm{CA^{h-1}(A-KC)\Sigma_{\hat X}} \Bigg)\abs{CA^{h-1}K - (CK + K_B)^{h}}\!
% \end{align*}
% for any $\ell \geq 1$. For $\rho(A) < \abs{CK + K_B} \approx 1$, there is a regime of terms in the above sum for which $\rho(A^h) \ll 1$ but $\abs{CK + K_B}^{h-1} \approx 1$. Consequently, the single-step predictor incurs excess error over what the multi-step predictor achieves that grows approximately linearly with the horizon for moderate prediction horizons $H$. \Bruce{Not true. Replace with numerical example}

% \Bruce{Looking at $\Gamma$, $KB$ and $G^\star$, it seems that some of these quantities may simplify nicely. If not in general, then at least in the scalar case. }
\section{Control Performance}
\label{sec: control}
We evaluate the predictor classes in a closed-loop control setting in which the control inputs are selected using predictions from a model trained on a dataset of size $N$. In particular, the control input is selected via model predictive control with a planning horizon equal to the predictor's forecast horizon $H$. Given an $H$-step predictor $\hat G$, we compute an $H$-step feedback gain $K_H(\hat G)$ such that 
\begin{align*}
    \bmat{ u_t \\ \vdots \\ u_{t+H-1}} = K_H(\hat G) y_t
\end{align*}
minimizes the finite horizon quadratic cost
\begin{align*}
    &\sum_{i=1}^H \norm{\hat y_{t+i}}^2 + \norm{u_{t+i-1}}^2 \\
    &\text{s.t. } \hat y_{t+1:t+H} = \hat G \bmat{y_t \\ u_{t:t+H-1}}.
    %&\phantom{\text{s.t. }} \hat y_{t+H} = 0.
\end{align*}
At execution time, controls are applied one step at a time according to $u_t = K(\hat G) y_t$ where $K(\hat G)$ is the first $d_u$-block row of $K_H(\hat G)$. The infinite-horizon LQR cost incurred by the resulting closed-loop system
\begin{align*}
    J  (K (\hat G)) = \limsup_{T \to \infty} \E\left [\sum_{t=1}^T  y_t^\top  y_t + u_t^\top u_t\right ].
\end{align*} 
gives a measure of performance of the predictors. However, on the event that $K(\hat G)$ is not stabilizing, $J(K(\hat G))$ will be infinite-valued. To account for this, we fix a large upper bound $M \gg 0$ and evaluate performance of predictors according instead to the clipped loss
\begin{align*}
    \tilde J(K(\hat G) = -\log \paren{\exp \{-J(K(\hat G))\} + \exp\{-M\}}.
\end{align*}

\subsection{Well-Specified Setting}
In the well-specified setting, we can compare the decay rates of the clipped infinite horizon LQR costs of controllers learned from a single-step predictor with a single step loss and a single step predictor with a multistep loss.
\begin{proposition}
\label{prop: LQR decay well-specified}

Assume $H$ is sufficiently large so that $\rho\paren{A + BK(G^\star)} < 1$. There exist constants $c\in \R$ and $\rho \in (0,1)$ such that 
\begin{align*}
&\lim_{N\to\infty} N\,\E\!\left[\abs{\tilde J\!\left(K(\hat G_N^{\mathrm{SS}})\right)-\tilde J\!\left(K(G^\star)\right)}\right] \\
&\le \lim_{N\to\infty} N\,\E\!\left[\abs{\tilde J\!\left(K(\hat G_N^{\mathrm{I}})\right)-\tilde J\!\left(K(G^\star)\right)}\right] + c\rho^H.
\end{align*}
\end{proposition}
\begin{proof}
    Fix $\epsilon>0$ so that $\rho\paren{A + BK(G^\star)} \leq 1-\epsilon$. There exists $r > 0$ such that if $\norm{K(G) - K(G^\star)} < r$, then $\rho \paren{A + BK} \leq 1 - \frac{\epsilon}{2}$. Define the event 
    \begin{align*}
        \calE (G) \triangleq \left \{ \norm{K(G) - K(G^\star)} < r \right \}. 
    \end{align*}
Then, for any predictor $\hat G_N$,
\begin{align*}
    &\lim_{N\to\infty} N\,\E\!\left[\abs{\tilde J\!\left(K(\hat G_N)\right)-\tilde J\!\left(K(G^\star)\right)}\right] \\
    &= \lim_{N\to\infty} N\,\E\!\left[\abs{\tilde J\!\left(K(\hat G_N)\right)-\tilde J\!\left(K(G^\star)\right)} \;\middle|\; \calE(\hat G_N)\right] \\
    &+ \lim_{N\to\infty} N\,\E\!\left[\abs{\tilde J\!\left(K(\hat G_N)\right)-\tilde J\!\left(K(G^\star)\right)} \;\middle|\; \calE(\hat G_N)^C\right]
\end{align*}
Note that 
\begin{align*}
    &\lim_{N\to\infty} N\,\E\!\left[\abs{\tilde J\!\left(K(\hat G_N)\right)-\tilde J\!\left(K(G^\star)\right)} \;\middle|\; \calE(\hat G_N)^C\right] \\
    &\leq \lim_{N\to\infty} 2NM \P\paren{\calE(\hat G_N)^C}. 
\end{align*}
From Theorem 5.3 of \cite{ziemann2023tutorial} , $\P\paren{\calE(\hat G_N^{SS})^C} \sim o\paren{\frac{1}{N}}$, so for the single step predictor, the complement term vanishes asymptotically. Thus, it suffices to compare the first term
\begin{align*}
    \lim_{N\to\infty} N\,\E\!\left[\abs{\tilde J\!\left(K(\hat G_N)\right)-\tilde J\!\left(K(G^\star)\right)} \;\middle|\; \calE(\hat G_N)\right]
\end{align*}
for $\hat G_N = \hat G_N^{SS}, \hat G_N^{I}$, and $\hat G_N^{MS}$. On $\calE(\hat G_N)$, $\tilde J\!\left(K(\hat G_N)\right)$ is smooth, so from \Cref{lemma: prop iv.1}, the expectation above can be written as 
\begin{align*}
    \trace \Bigl(\nabla_{\VEC (G)} K(G^\star)^T\nabla^2_{\VEC (K)} \tilde J(K^\star) \nabla_{\VEC (G)} K(G^\star) \Sigma_{\hat G_N}\Bigl)
\end{align*}
plus a term which scales as $\rho^H$ for $\rho \in (0,1)$, where $\Sigma_{\hat G_N}$ is the asymptotic variance 
\begin{align*}
    \lim_{N \rightarrow \infty} N \var(\VEC (\hat G_N - G^\star)) 
\end{align*} 
and $K^\star = \argmin_K J(K)$.
By the Cramér Rao Lower Bound, we have that the asymptotic variance of the single step predictor $\hat G_N^{SS}$ is smaller than that of the intermediate predictor in the p.s.d cone. 
Since 
\begin{align*}
    \nabla_{\VEC (G)} K(G^\star)^T\nabla^2_{\VEC (K)} \tilde J(K(^\star) \nabla_{\VEC (G)} K(G^\star) \succeq 0,
\end{align*}  the final result holds. 
\end{proof}
The above proposition states that the infinite horizon LQR cost associated with a single-step predictor decays faster than that of the intermediate predictor up to a term which decays exponentially in the horizon $H$. Comparison with the multi-step predictor is left for future work. 

\begin{figure}

        \centering
        \includegraphics[width=\linewidth]{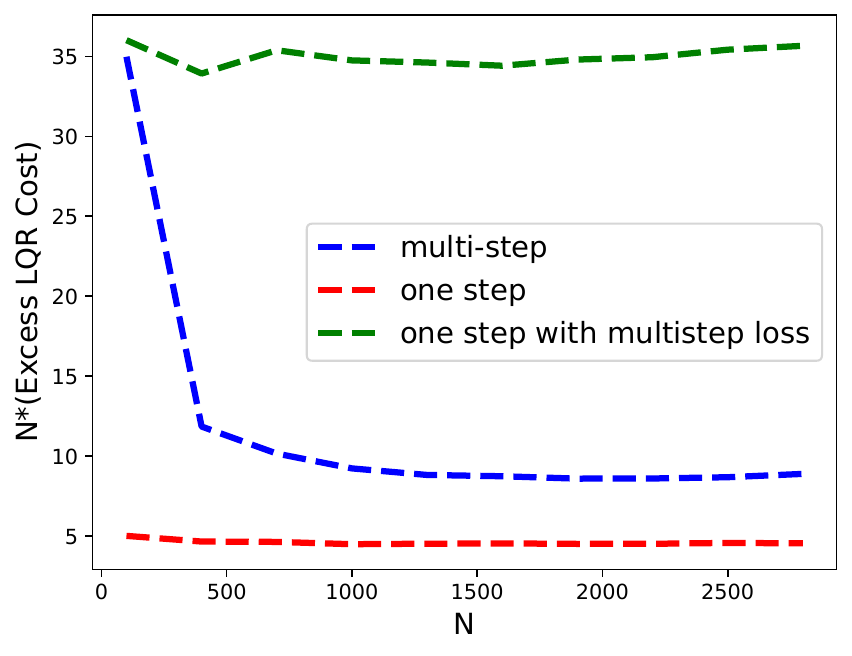}
        \label{fig:well-specified}

    \caption{Infinite-horizon LQR performance in the well-specified case. }

    %Infinite horizon LQR cost 
    \label{fig: LQRcost}
\end{figure}

In \Cref{fig: LQRcost}, we consider the closed-loop control setting above for system \eqref{eq: numerical ex system well specified} with $a = 0.9$ and $H=20$ averaged over $2,500$ datasets $D_N$ for $N \in \{ 1,...,3000\}$. As suggested by \Cref{prop: LQR decay well-specified}, the single-step predictor exhibits a faster decay rate than the intermediate predictor. However, though the intermediate predictor achieves a faster prediction error decay rate than the multi-step predictor, as shown in \cref{prop: predictor ordering well specified}, it exhibits a slower LQR cost decay rate than the multi-step in closed loop.

\subsection{Misspecified Setting}
In the previous section, we analyzed the infinite-horizon control cost induced by controllers designed using multi-step predictions obtained from (a) a single-step predictor, (b) a single-step predictor trained with a multi-step loss, and (c) a direct multi-step predictor. Under model misspecification, however, irreducible prediction bias may prevent any static state-feedback controller derived from these predictors from stabilizing the true system. In such cases, the infinite-horizon control cost is infinite, even in the limit as the dataset size $N \to \infty$. Consequently, rather than comparing asymptotic control costs, we instead compare predictors through their ability to induce stabilizing controllers. In particular, we identify regimes in which the lower-bias multi-step predictor yields a stabilizing controller, whereas the alternative predictors may fail to do so. We provide numerical examples illustrating this phenomenon and leave theoretical characterization of these regimes to future work.

\begin{figure} 
    \centering
    \begin{subfigure}{0.9\linewidth}
        \centering
        \includegraphics[width=\linewidth]{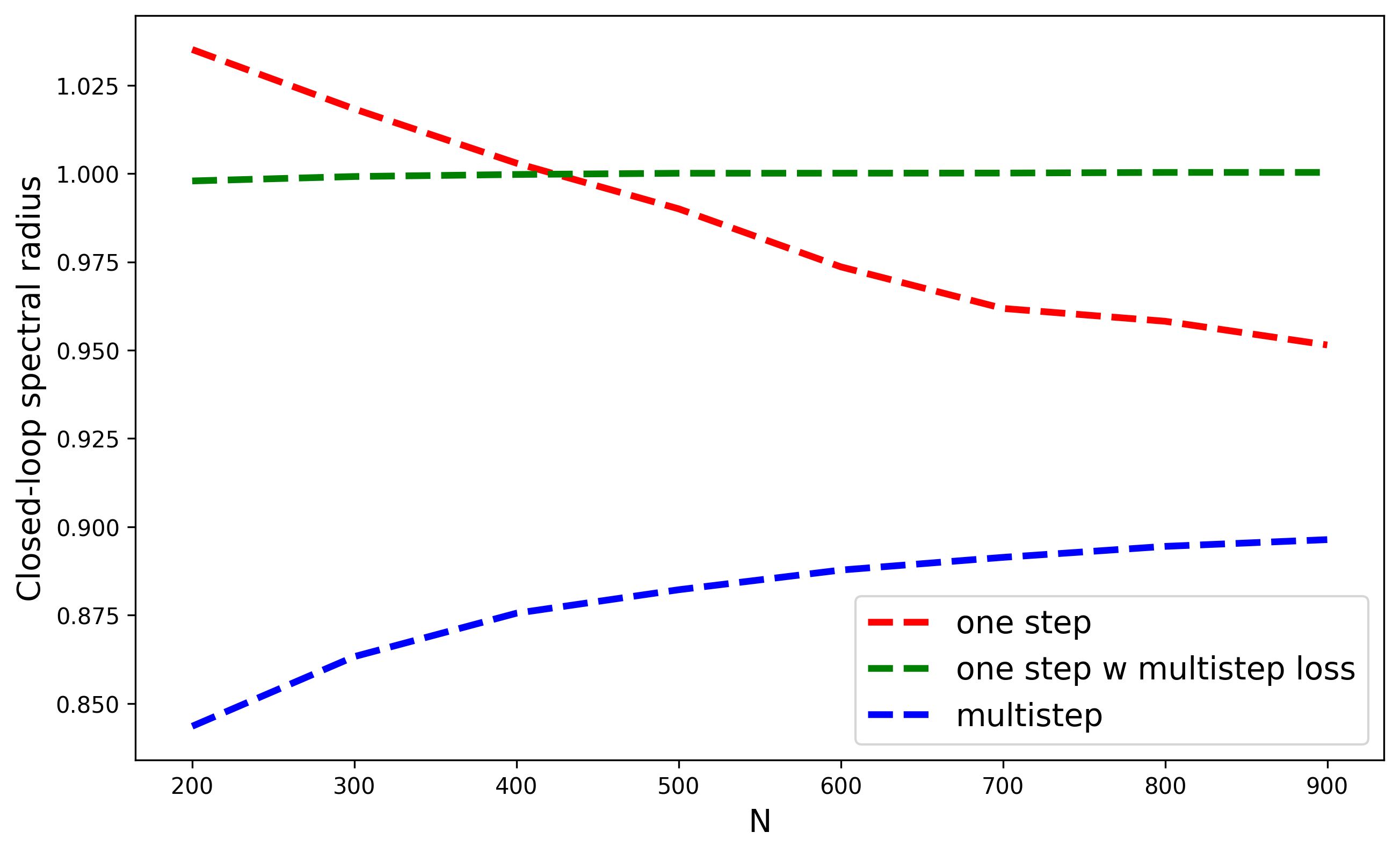}

        \caption{Closed loop spectral radius for system \eqref{eq: numerical ex system misspecified} with $a = 0.6$. }  
    \end{subfigure}
    \begin{subfigure}{0.9\linewidth}
        \centering
        \includegraphics[width=\linewidth]{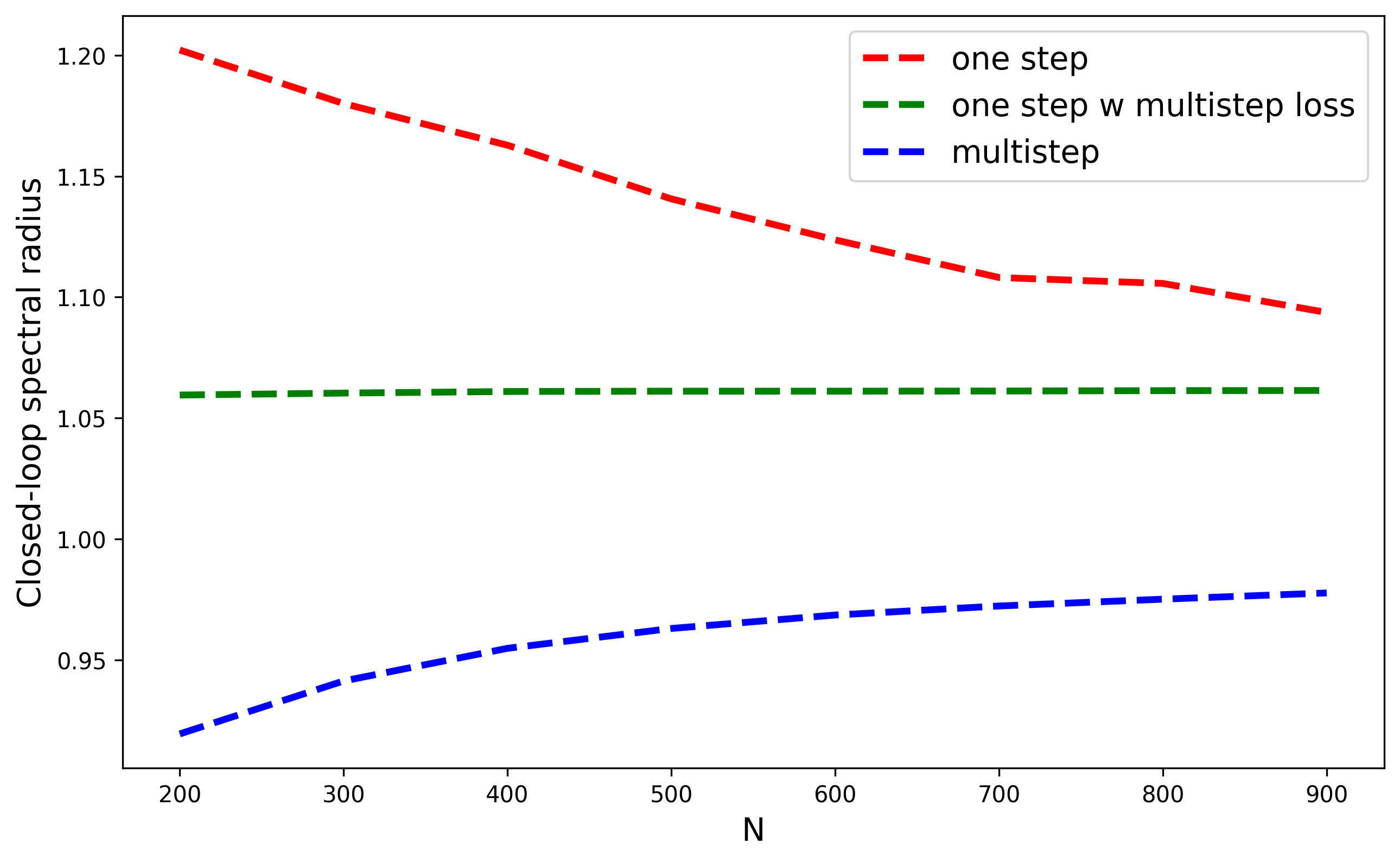}

        \caption{Closed loop spectral radius for system \eqref{eq: numerical ex system misspecified} with $a = 0.75$. }  
    \end{subfigure}
    \caption{Infinite-horizon LQR performance in the misspecified case. Closed loop spectral radius greater than 1 implies infinite LQR cost. }
    \label{fig: spectral radius}
\end{figure}

\Cref{fig: spectral radius} shows the spectral radius of the resulting closed-loop system in the misspecified setting for the system \eqref{eq: numerical ex system misspecified} with $H = 20$ averaged over 2,500 datasets. A spectral radius greater than 1 for the single-step rollout indicates closed-loop system instability, in which case the associated infinite-horizon LQR cost diverges. Panel (a) corresponds to $a = 0.6$, where both the single step and intermediate predictors yield stabilizing controllers but, for small $N$, the single step predictor does not. Panel (b) corresponds to $a = 0.75$, where stabilization is achieved only by the multi-step predictor. These results suggest that controllers induced by direct multi-step predictors are more likely to stabilize the true system than those obtained from single-step rollouts. We provide two explanations for this.

First, let $K^\star$ denote the feedback gain computed from the ground truth multi-step predictor $G^\star$. For sufficiently large planning horizon $H$, the associated closed-loop system is stable, i.e. $\rho(A + BK^\star) <1$. By continuity of the spectral radius, it follows that if a feedback gain $K$ has $\norm{K - K^\star}$ sufficiently small, then $K$ yields a stabilizing closed-loop matrix. Since $K(\hat G)$ is a continuous function of $\hat G$, it then follows that for $\norm{\hat G - G^\star}$ sufficiently small, $\rho \paren{A + BK(\hat G)} <1$. In \Cref{prop: predictor ordering misspecified}, we showed that, in the presence of misspecification, direct multi-step predictors achieve lower asymptotic prediction error than the other predictor types. Thus, they are more likely to lie within this stabilization neighborhood. 

Second, for stable systems, multi-step prediction errors are dominated by errors in the direction of the modes with the largest magnitude. Since these directions are those which govern stabilization, training predictors with a multi-step loss implicitly emphasizes accurate modeling in the directions which are most relevant for stabilization. Thus, multi-step predictors are better suited for inducing stabilizing controllers. 
\section{Nonlinear Systems}
\label{sec: nonlinear}

In Sections \ref{sec: well-specified} and \ref{sec: misspecified}, we provided a rigorous comparison of the prediction error incurred by single- and multi-step predictors. Though these results are restricted to linear systems, we provide empirical evidence that similar phenomena arise for nonlinear dynamics. Consider the system with state $x_t = \bmat{p_t \ q_t}$, initialized at $x_0 = \bmat{0 \ 0}$, and evolving as
\begin{equation}
\label{eq: nonlinear}
\begin{aligned}
p_{t+1} &= \mu p_t + w_t^{(p)}, \\
q_{t+1} &= \lambda (q_t - p_t^2) + w_t^{(q)},
\end{aligned}
\end{equation}
where $w_t^{(p)}, w_t^{(q)} \overset{\iid}{\sim} \calN(0,\sigma_w)$.

This system admits a finite-dimensional Koopman lifting (e.g. \cite{korda2018linear}) with observables $\tilde x_t = \bmat{p_t \ q_t \ p_t^2 \ 1}$, for which the dynamics are linear in $\tilde x_t$ up to noise terms:
\begin{equation}
    \label{eq: koopman}
    \begin{aligned}
        \tilde x_{t+1} &= \bmat{\mu & 0 & 0 & w_t^{(p)} \\ 0 & \lambda & -\lambda & w_t^{(q)} \\ 2\mu w_t^{(p)} & 0 & \mu^2 & \paren{w_t^{(p)}}^2 \\ 0 & 0 & 0 & 1}\tilde x_t.
    \end{aligned}
\end{equation}

We consider observations
\begin{equation}
\label{eq: koopman observations}
y_t = C \tilde x_t + v_t, \quad v_t \overset{\iid}{\sim} \calN(0,\sigma_v I_{d_y}),
\end{equation}
and, given a dataset $\calD_N = \{y_t \}_{t=1}^N$, compute the single-step predictor $\hat G_N^{SS}$, multi-step predictor $\hat G_N^{MS}$, and intermediate predictor $\hat G_N^{I}$ as described in \Cref{s:singlestep,s: multi-step,s: intermediate}.

We evaluate these predictors in both a well-specified setting ($C = I$, $\sigma_v = 0$) and a misspecified setting ($C \neq I$, $\sigma_v > 0$), analogous to \Cref{sec: well-specified,sec: misspecified}.

\begin{figure}
\centering
\begin{subfigure}{0.9\linewidth}
\centering
\includegraphics[width=\linewidth]{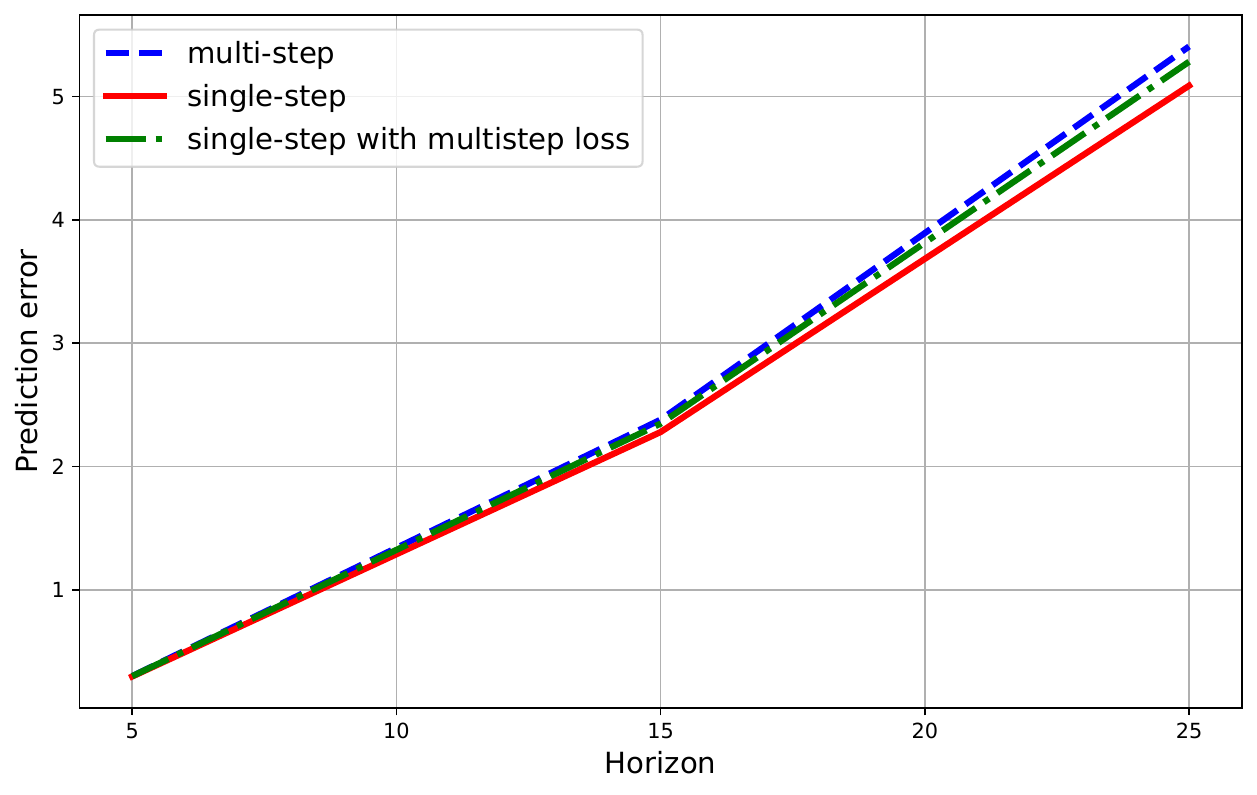}
\caption{Well-specified setting.}
\label{fig: koopman well-specified}
\end{subfigure}
\begin{subfigure}{0.9\linewidth}
\centering
\includegraphics[width=\linewidth]{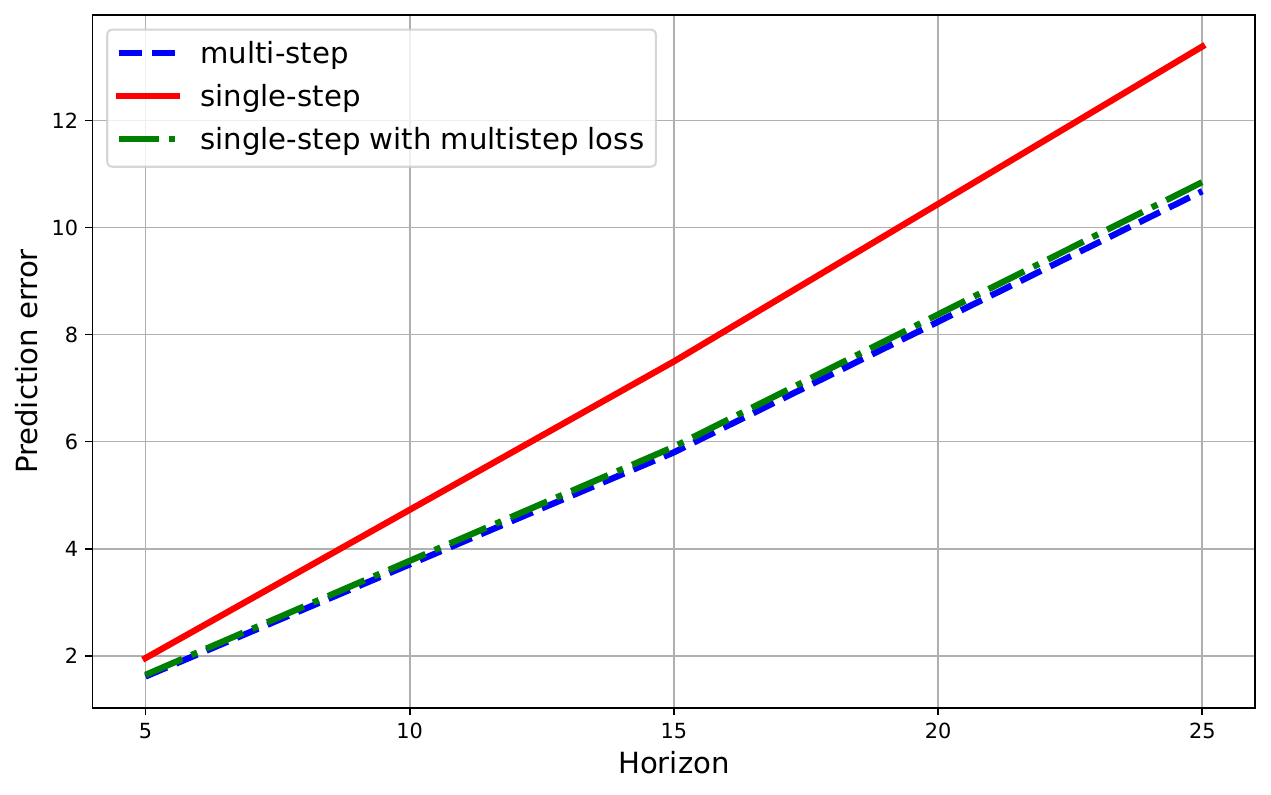}
\caption{Misspecified setting.}
\label{fig: koopman misspecified}
\end{subfigure}
\caption{Prediction error versus horizon $H$ for the system \eqref{eq: koopman observations}.}
\label{fig: nonlinear}
\end{figure}

In \Cref{fig: nonlinear}, we estimate expected mean-squared prediction error by averaging over $200$ datasets with $N=300$. We fix $\mu=\lambda=0.9$, $\sigma_w=0.1$, and vary $H \in [5,25]$. In the misspecified case, we use $C=\bmat{0 & 1 & 0 & 0}$ and $\sigma_v=0.4$. The intermediate predictor $\hat G_N^{I}$ is trained using Adam (step size $10^{-2}$), initialized from $\hat G_N^{SS}$.

The results mirror our linear theory. In the well-specified setting, \Cref{fig: koopman well-specified} shows that the single-step predictor achieves the lowest error, followed by the single-step predictor trained with a multi-step loss, and then the direct multi-step predictor, consistent with \Cref{prop: predictor ordering well specified}. In contrast, \Cref{fig: koopman misspecified} shows the reverse ordering, consistent with \Cref{prop: predictor ordering misspecified}.

These experiments suggest that the qualitative behavior predicted by our linear analysis may extend to some nonlinear systems. A rigorous treatment of this setting is left for future work.

\section{Conclusion}
In this work, we present a novel theoretical comparison of the asymptotic prediction error associated with autoregressive rollouts of single-step predictors, direct multi-step predictors, and single-step predictors fitted with multi-step losses. Our analysis offers insight into when each modeling approach is preferable. Specifically, we show that for well-specified model classes, autoregressive rollouts of single-step predictors achieve lower asymptotic prediction error. However, in the presence of model misspecification due to an incorrect Markovian assumption, multi-step predictors can significantly outperform their single-step counterparts.

These findings provide a foundation for more informed model design in learning-based control and forecasting. Promising directions for future work include: (1) extending these results to the setting of nonlinear systems and (2) analyzing predictor performance in closed-loop control in frameworks other than the LQR setting studied in this work, e.g. model based reinforcement learning.
\appendices

\section*{Acknowledgements}
This work is supported in part by NSF Award SLES-2331880, NSF CAREER award ECCS-2045834, and AFOSR Award FA9550-24-1-0102.

\bibliographystyle{IEEEtranN}
\bibliography{refs}

\section{Proofs of \cref{sec: well-specified} Results}
\label{appendix: well-specified}

\subsection{Proof of \Cref{prop: well specified multistep}}

\begin{lemma} \label{lemma: 1 prop2.1} For $\Gamma_w$ as defined in \Cref{sec: well-specified} and $M_{MS}$ as in \Cref{prop: well specified multistep}, 
\begin{align*}
&\lim_{N\to\infty}
\frac{1}{N}\,\E\Biggl\|
\sum_{t=1}^{N-H+1}
\Gamma_w w_{t:t+H-1} z_t^\top \Sigma_z^{-1/2}
\Biggr\|_F^2  \\
&=\trace\paren{\Gamma_w ((M_{MS} + H\du I_H) \otimes I_{\dx}) \Gamma_w^\top}.
\end{align*}

\end{lemma}
\begin{proof}
Expanding the Frobenius norm results in the trace
\begin{align*}
\sum_{t,k=1}^{N-H+1} \trace\paren{ \Gamma_w w_{t:t+H-1} z_t^\top \Sigma_z^{-1} z_k w_{k:k+H-1}^\top \Gamma_w^\top}. 
\end{align*}
Consider the expectation of the summand for any pair of indices $t$ and $k$. If $t = k$, we may use the fact that $w_{t:t+H-1}$ is independent from $z_t$ to conclude
\begin{align*}
    &\E\brac{\trace\paren{ \Gamma_w w_{t:t+H-1} z_t^\top \Sigma_z^{-1} z_k w_{k:k+H-1}^\top \Gamma_w^\top}} \\ 
    &= \trace\paren{\Gamma_w\E \brac{w_{t:t+H-1} \E \brac{z_t^\top \Sigma_z^{-1} z_t} w_{t:t+H-1}^\top} \Gamma_w^\top} \\
    &= \trace\paren{\Gamma_w (\dx + H \du)  \Gamma_w}.
\end{align*}
If instead $t = k +m$ for $m > 0$, we have the following term 
\begin{align*}
    \E\brac{\trace\paren{ \Gamma_w w_{k+m:k+m+H-1} z_{k+m}^\top \Sigma_z^{-1} z_k w_{k:k+H-1}^\top \Gamma_w^\top}}.
\end{align*}
Expanding $z_{k+m}$ provides
\[
    z_{k+m} = \bmat{A^m x_k + \sum_{\ell=0}^{m-1} A^{m-1-\ell} (B u_{k+\ell} + B_w w_{k+\ell}) \\ u_{k+m: k+m+H-1}},
\]
and thus the expectation simplifies to
\begin{align*}
     &\E\brac{\trace\paren{ \Gamma_w w_{t:t+H-1} z_t^\top \Sigma_z^{-1} z_k w_{k:k+H-1}^\top \Gamma_w^\top}}.  
\end{align*}
Taking the limit, 
\begin{align*}
    &\lim_{N\to\infty} \frac{1}{N}\E\brac{\trace\paren{ \Gamma_w w_{t:t+H-1} z_t^\top \Sigma_z^{-1} z_k w_{k:k+H-1}^\top \Gamma_w^\top}} \\
    &= \trace\paren{\Gamma_w \trace(A^m) (U^m \kron I_{\dx})\Gamma_w^\top}, 
\end{align*}
where $U$ is the upper shift matrix. The case $m<0$ can be shown to reduce to a similar quantity with $U$ replaced by the lower shift matrix $U^\top$. Summing over all indices $m$ and combining terms gives the final result.  

\end{proof}

\subsection{Proof of \Cref{prop: single-step well specified}} \label{appendix: single step well specified}

The following definitions are used in the proof of \Cref{prop: single-step well specified}. Define
\begin{align*} 
F
&\triangleq
\bmat{
I_{\dx} \\
& I_{\du} \\
A & B \\
& & I_{\du} \\
\vdots \\
A^{H-1} & A^{H-2}B & \dots & B \\
& & & & I_{\du}
},
\\[0.5em]
\Gamma
&\triangleq
\bmat{
I_{\dx} \\
A & I_{\dx} \\
\vdots \\
A^{H-1} & A^{H-2} & \dots & I_{\dx}
}.
\end{align*}

% Define \(\Gamma \in \R^{Hd_x \times Hd_x}\) blockwise by
% \[
% \Gamma_{ij} =
% \begin{cases}
% A^{i-j}, & i \ge j,\\
% 0, & i < j,
% \end{cases}
% \qquad i,j=1,\dots,H.
% \]

\begin{lemma} \label{lemma: ss single step var}
Let $\bmat{ \hat G_y & \hat G_u}$ be the first row of $\hat G_N^{SS}$ as defined in \eqref{eq: single-step rolled out G}. 
\begin{align*}
    \lim_{N \rightarrow \infty}N \var \left ( \VEC \left(\bmat{\hat G_y \!-\! A\! & \!\!\hat G_u \!-\! B}\right) \right) =  \Sigma_{x,u}^{-1} \otimes B_w B_w^\top 
\end{align*}
where $\Sigma_{x,u}$ is the stationary covariance of $\bmat{x_t^\top & u_t^\top}^\top$.
\end{lemma}
\begin{proof}

By the normal equations for the least squares estimator, 
\begin{align*} 
    &\VEC \left(\bmat{\hat G_y \!-\! A\! & \!\!\hat G_u \!-\! B}\right ) \\
    &= \left(\left(\sum_{t=1}^{N-1} \bmat{x_t \\ u_t} \bmat{x_t \\ u_t}^\top \right)^{-1} \otimes B_w\right )\left (\sum_{t=1}^{N-1} \left (\bmat{x_t \\ u_t} \otimes I_{d_x}\right)w_t \right ).
\end{align*}
A combination of the Birkoff-Khinchin theorem and Slutsky's theorem shows that 
\begin{align}\label{eq: proof II.2 one step var}
    &N\VEC \left(\bmat{\hat G_y \!-\! A\! & \!\!\hat G_u \!-\! B}\right)^\top \VEC \left(\bmat{\hat G_y \!-\! A\! & \!\!\hat G_u \!-\! B}\right) \notag\\
    &\rightarrow \Sigma_{x,u}^{-1} \otimes B_w B_w^\top .
\end{align}

\end{proof}
% We can show using (Ljung?) that 
% \begin{align*}
%     &\sqrt{N}\VEC (\bmat{\hat G_y - G_y & \hat G_u - G_u}  )  \\
%     &\quad \quad \quad \in As{\calN}(0, B_w^2(\Sigma_{Z,1}^{-1} \otimes I_{\dx})).
% \end{align*} 

\subsection{Proof of \Cref{prop: ss w/ ms loss well specified}}
\begin{lemma} \label{lemma: ss_ms var}
Let $\bmat{ \hat G_y & \hat G_u}$ be the first row of $\hat G_N^I$ as defined in \eqref{eq: multistep LS}. 
\begin{align*}
\lim_{N\to\infty}
N\,\var\!\left(\VEC\!\left(
\begin{bmatrix}
\hat G_y - A & \hat G_u - B
\end{bmatrix}
\right)\right)
= J^{-1}\Sigma J^{-1}.
\end{align*}
where $J$ is defined in \eqref{eq: hessian} and $\Sigma$ is defined in \eqref{eq: jacobian}. 
\end{lemma}
\label{proof: ss w/ ms loss well-specified}
\begin{proof}
Recall the definition of $m_t(\cdot, \cdot)$ from \eqref{eq: per timestep loss} and $\theta$ from \eqref{eq: theta}. 
% Define 
% \begin{align*}
%     &m_t(G_y, G_u) \\ &\coloneq \sum_{k=1}^H \lVert y_{t+k} - \bmat{G_y^k & G_y^{k-1}G_u & ... & G_u} \bmat{y_t \\ u_{t:t+k-1}} \rVert^2
% \end{align*}
It holds that
\begin{align*}
    &\hat G_y, \hat G_u \in  \argmin_{G_y, G_u} \frac{1}{N} \sum_{t=1}^{N-H} m_t(G_y, G_u).
\end{align*}
Since $\hat G_y, \hat G_u$ is a minimizer, 
\begin{align*}
    \sum_{t=1}^{N-H}\nabla_{\theta}  m_t(\hat G_y, \hat G_u) = 0.
\end{align*}
Taylor expanding: 
\begin{align*}
    0 &= \frac{1}{N}\sum_{t=1}^{N-H}\nabla_{\theta} m_t(\hat G_y, \hat G_u) \\& = \frac{1}{N}\sum_{t=1}^{N-H}\nabla_{\theta} m_t(A,  B) \\&+ \frac{1}{N}\sum_{t=1}^{N-H}\nabla^2_{\theta} m_t(A, B)  \cdot \VEC \Bigl (\bmat{\hat G_y - A & \hat G_u - B} \Bigl) \\& + \frac{1}{N}\calO \Bigl( \lVert \bmat{\hat G_y - A & \hat G_u - B} \rVert^2 \Bigl).
\end{align*}
Rearranging, 
\begin{align*}
    &\sqrt{N} \VEC \Bigl (\bmat{\hat G_y - A & \hat G_u - B} \Bigl) = -\Biggl(\frac{1}{N} \sum_{t=1}^{N-H}\nabla^2_{\theta} m_t(A,  B) \Biggl)^{-1} \\
    & \cdot \Biggl(\frac{1}{\sqrt{N}} \sum_{t=1}^{N-H}\nabla_{\theta} m_t(A,  B)  + \frac{1}{\sqrt{N}}\calO \Bigl( \lVert \bmat{\hat G_y - A & \hat G_u - B} \rVert^2 \Bigl)\Biggl).
\end{align*}

Taking an outer product, we see that
\begin{align*}
    &N \VEC \Bigl (\bmat{\hat G_y - A & \hat G_u - B} \Bigl)\VEC \Bigl (\bmat{\hat G_y - A & \hat G_u - B} \Bigl)^\top\\
    &\approx \!-\!\Biggl(\frac{1}{N} \sum_{t=1}^{N-H}\nabla^2_{\theta} m_t(A,  B) \Biggl)^{-1} \!\!\Biggl(\!\frac{1}{\sqrt{N}}\! \sum_{t=1}^{N-H}\nabla_
    {\theta} m_t(A,  B)\! \Biggl) \\
    &\times\!\Biggl(\!\frac{1}{\sqrt{N}} \!\sum_{t=1}^{N-H}\nabla_{\theta} \!m_t(A, B) \Biggl)^T \Biggl(\!\frac{1}{N} \!\sum_{t=1}^{N-H}\!\nabla^2_{\theta} m_t(A,  B) \Biggl)^{-T},
\end{align*}
where the $\approx$ indicates that we've dropped higher order terms.
By ergodicity, 
\begin{align*}
    &\lim_{N \rightarrow \infty} \frac{1}{N} \sum_{t=1}^{N-H}\nabla^2_{\theta} m_t(A,  B) 
    % &= \bar \E \Bigl[\nabla^2_{\VEC \Bigl (\bmat{G_y & G_u} \Bigl)} m_t(G_y,  G_u) \Bigl] \\
     = J,
\end{align*}
for $J$ from \eqref{eq: hessian}
and 
\begin{align*}
    & \lim_{N \rightarrow \infty} \Biggl(\frac{1}{\sqrt{N}} \sum_{t=1}^{N-H}\nabla_{\theta} m_t(A,  B) \Biggl) \\
    &\cdot\Biggl(\frac{1}{\sqrt{N}} \sum_{t=1}^{N-H}\nabla_{\theta} m_t(A, B) \Biggl)^T \\
    & = \lim_{N \rightarrow \infty} \frac{1}{N} \sum_{t=1}^{N-H} \sum_{\tau = t}^{N-t} \nabla_{\theta} m_t(A,  B) \nabla_{\theta} m_{t+\tau}(A,  B) \top \\
    % &= \sum_{\tau = -t}^\infty \bar \E \Bigl[ \nabla_{\VEC \Bigl (\bmat{G_y & G_u} \Bigl)} m_t(G_y,  G_u)\\
    % & \cdot \nabla_{\VEC \Bigl (\bmat{G_y & G_u} \Bigl)} m_{t+\tau}(G_y,  G_u) ^\top\Bigl] \\
    & = \Sigma,
\end{align*}
for $\Sigma$ in \eqref{eq: jacobian}.
Recall that $\bar \E$ is equivalent to taking an expectation under the stationary distribution and therefore $J$ and $\Sigma$ do not depend on $t$. All in all, 
\begin{align*}
    & \lim_{N \rightarrow \infty} N \mathsf{var}\paren{ \VEC \Bigl (\bmat{\hat G_y & \hat G_u } \Bigl)} = J^{-1} \Sigma J^{-1}.
\end{align*}

\end{proof}

\subsection{Proof of \Cref{prop: predictor ordering well specified}}
\label{proof: predictor ordering well specified}
\begin{lemma} \label{lemma: hansen}
    Let $V_{\hat \beta^{MS}}, V_{\hat \beta^{I}}$, and $Q$ be defined as in Section~II-A. It holds that 
\begin{align*}
    \trace(V_{\hat \beta^{I}} Q) = \trace(V_{\hat \beta^{MS}} Q) - \trace(V_{\hat \beta^{MS}}^{\frac{1}{2}}R(R^\top Q^{-1}R)^{-1}R^\top V_{\hat \beta^{MS}}^{\frac{1}{2}})
\end{align*}
for some $R \in \R^{H\dy(\dy + H\du) \times (H-1)\dy(\dy + H\du)}$ with full row rank.  
\end{lemma}
\begin{proof}

Let $\bmat{\hat G_y & \hat G_u}$ be the first block row of the intermediate predictor $\hat G_N^{I}.$ That is, 
\begin{align*}
    &\hat G_y, \hat G_u \in \arg \min_{G_y, G_u} L(G_y, G_u)
\end{align*}
where
\begin{align*}
    &L(G_y, G_u) \triangleq \\
    &\sum_{t=1}^{N-H} \norm{ \bmat{x_{t+1}\\ \vdots \\ x_{t+H}} - \bmat{G_y & G_u \\ G_y^2&  G_yG_u & G_u \\ &\vdots && \ddots \\ G_y^H & G_y^{H-1}G_u && \dots & G_u} z_t }^2.
\end{align*}
Then, 
\begin{align*}
    \hat \beta^{I} = \VEC \left (\bmat{\hat G_y & \hat G_u \\ \hat G_y^2&  \hat G_y \hat G_u & \hat G_u \\ &\vdots && \ddots \\ \hat G_y^H & \hat G_y^{H-1}\hat G_u && \dots & \hat G_u} \right).
\end{align*}
Rolling out and vectorizing, our system dynamics can be written
\begin{align*}
    x_{t+1:t+H} = (z_t^\top \kron I_{Hd_x})\VEC(G^\star) + \Gamma_w w_{t:t+H-1}.
\end{align*}  Letting $\beta = \VEC \left( \bmat{G_y & G_u \\ G_y^2&  G_yG_u & G_u \\ &\vdots && \ddots \\ G_y^H & G_y^{H-1}G_u && \dots & G_u}\right)$, we see that 
\begin{align*}
    &L(G_y, G_u) = \sum_{t=1}^{N-H} \lVert (z_t^T \kron I_{Hd_x}) \hat \beta^{MS} + \hat e_t - (z_t^T \kron I_{Hd_x}) \beta \rVert^2
\end{align*}
where 
\begin{align*}
    \hat e_t = \bmat{x_{t+1}\\ \vdots \\ x_{t+H}} - (z_t^T \kron I_{Hd_x}) \hat \beta^{MS}
\end{align*}
are the residuals from the unconstrained multi-step predictor. The above sum can be simplified to 
\begin{align*}
    \sum_{t=1}^{N-H} (\hat \beta^{MS} - \beta)^\top (z_tz_t^\top \kron I_{Hd_x})(\hat \beta^{MS} - \beta) + \hat e_t^\top \hat e_t. 
\end{align*}
Let 
$
    \hat Q_N = \sum_{t=1}^{N-H}z_tz_t^\top \kron I_{Hd_x}
$
and note that $\frac{1}{N} \hat Q_N \rightarrow Q .$
Then, the loss becomes
\begin{align*}
    (\hat \beta^{MS} - \beta)^\top \hat Q_N (\hat \beta^{MS} - \beta) + \sum_{t=1}^{N-H} \hat e_t^\top \hat e_t
\end{align*}
and the constrained multi-step predictor $\hat \beta^{I}$ can be written equivalently as the solution of the problem 
\begin{align} \label{cls}
    &\min_\beta (\hat \beta^{MS} - \beta)^\top \hat Q_N (\hat \beta^{MS} - \beta) \\
    & \text{s.t.} \quad r(\beta) = 0
\end{align}
where the function $r: \R^{H\dy(\dy + H\du)} \rightarrow \R^{(H-1)\dy^2 + (H^2-1)\dy\du}$ is defined so that $r(\beta) = 0$ encodes the constraint that if $\beta = \VEC(G)$ for $G \in \R^{Hd_y \times d_y + Hd_u} $ then $G$ has the form 
\begin{align*}
    \bmat{G_y & G_u \\ G_y^2&  G_yG_u & G_u \\ &\vdots && \ddots \\ G_y^H & G_y^{H-1}G_u && \dots & G_u}.
\end{align*}
for some $G_y \in \R^{\dx \times \dx}, G_u \in \R^{\dx \times \du}$. Specifically, if $G_y = G_{0:\dy, 0:\dy}$ and $G_u = G_{0:\dy, \dy:\dy + \du}$ then $r$ is the vectorized form of the map 
\begin{align*}
    G \rightarrow G - \bmat{G_y & G_u \\ G_y^2&  G_yG_u & G_u \\ &\vdots && \ddots \\ G_y^H & G_y^{H-1}G_u && \dots & G_u}
\end{align*}
where we disregard the elements in the top left $\dy \times \dy+\du$ block so that the Jacobian of $r$ has full row-rank. 
Denote the components of $r$ by
\begin{align*}
    r(\beta) = \bmat{r_1(\beta) \\ \vdots \\ r_q(\beta)}. 
\end{align*}
By the mean value theorem, for each component $r_j$ there exists a point $\beta_j$ on the line segment between $\hat \beta^{I}$ and $\beta^\star$ such that 
\begin{align*}
    r_j(\hat \beta^{I}) = r_j(\beta^\star) + \Bigl ( \frac{\partial r_j}{\partial \beta}(\beta_j)\Bigl)^\top (\hat \beta^{I} - \beta^\star). 
\end{align*}
Let 
\begin{align*}
    \hat R_N^{\star^\top} = \bmat{\frac{\partial r_1}{\partial \beta}(\beta_1)^\top \\
    \vdots \\
     \frac{\partial r_q}{\partial \beta}(\beta_q)^\top}
\end{align*} 
so that
$
    r(\hat \beta^{I}) = r(\beta^\star) + \hat R_N^{\star^\top}(\hat \beta^{I} - \beta^\star).
$
Since $r(\hat \beta^{I}) = r(\beta^\star) = 0$, this simplifies to the linearized constraint 
\begin{align*}
    \hat R_N^{\star^\top} (\hat \beta^{I} - \beta^\star) = 0
\end{align*}
Note that $ \hat R_N^\star \rightarrow R \coloneq \frac{\partial r(\beta^\star)}{\partial \beta}^\top.$
Equation \eqref{cls} has first order condition 
\begin{align*}
   \hat Q_N (\hat \beta^{MS} - \hat \beta^{I}) = \hat R \lambda
\end{align*}
where $\hat R = \frac{\partial r}{\partial \beta}(\hat \beta^{I})^\top \rightarrow R$. Multiplying by $\hat R_N^{\star ^\top} \hat Q_N^{-1}$ on the left and solving for $\lambda$ gives 
\begin{align*}
    \lambda = (\hat R_N^{\star^\top} \hat Q_N^{-1} \hat R)^{-1} \hat R_N^{\star^\top} (\hat \beta^{MS} - \hat \beta^{I}).
\end{align*}
Substituting this expression into the F.O.C. and rearranging, we can show that 
\begin{align} \label{pred_comp}
    \hat \beta^{I} - \beta^\star = (I - \hat Q_N^{-1}\hat R(\hat R_N^{\star^\top} \hat Q_N^{-1} \hat R)^{-1} \hat R_N^{\star^\top})(\hat \beta^{MS} - \beta^\star).
\end{align}
Let $V_{\hat \beta^{I}}$ denote the asymptotic variance of the constrained multi-step predictor 
$
    V_{\hat \beta^{I}} \coloneq \lim_{N \rightarrow \infty} N \text{Var}(\hat \beta^{I} - \beta).
$
From \eqref{pred_comp}, it follows that 
\begin{align*}
    V_{\hat \beta^{I}} &= V_{\hat \beta^{MS}} - Q^{-1}R(R^\top Q^{-1}R)^{-1} R^\top V_{\hat \beta^{MS}} \\
    &- V_{\hat \beta^{MS}} R (R^\top Q^{-1}R)^{-1}R^\top Q^{-1} \\
    &+ Q^{-1}R(R^\top Q^{-1}R)^{-1} R^\top V_{\hat \beta^{MS}}R (R^\top Q^{-1}R)^{-1}R^\top Q^{-1}.
\end{align*}
It follows directly that 
\begin{align*}
    \trace(V_{\hat \beta^{I}} Q) = \trace(V_{\hat \beta^{MS}} Q) - \trace(V_{\hat \beta^{MS}}^{\frac{1}{2}}R(R^\top Q^{-1}R)^{-1}R^\top V_{\hat \beta^{MS}}^{\frac{1}{2}}).
\end{align*}

\end{proof}

\section{ \cref{sec: misspecified} Additional Results}
\label{appendix: misspecified}

\subsection{Spectral radius of $C A \Sigma_{\hat x} C^\top \Sigma_y^{-1}$}
\begin{lemma}
\label{lemma: spectral radius of CK+KB}
    Assume $\rho(A) < 1$. Then, 
    \begin{align*}
        \rho(C A \Sigma_{\hat x} C^\top \Sigma_y^{-1}) \leq 1. 
    \end{align*}
\end{lemma}
\begin{proof}[Proof of \Cref{lemma: spectral radius of CK+KB}]
Note that $ C A \Sigma_{\hat x} C^\top \Sigma_y^{-1} = \bar \E[y_{t+1} y_t^\top] \bar \E[y_t y_t^\top]^{-1}$. From the stationarity of the process,
\begin{align*}
    \bar \E \bmat{y_t \\ y_{t+1}} \bmat{y_t \\ y_{t+1}}^\top = \bmat{\Sigma_y & \Sigma_{y+} \\ \Sigma_{y+} & \Sigma_y}. 
\end{align*}
By a Schur complement,  $\Sigma_y - \Sigma_{y+} \Sigma_y^{-1} \Sigma_{y+} \succeq 0,$ or $(\Sigma_y^{-1/2} \Sigma_{y+} \Sigma_{Y}^{-1/2})^2 \preceq I$. Then $\norm{\Sigma_y^{-1/2} \Sigma_{y+} \Sigma_{Y}^{-1/2}} \leq 1$. For $i=1,\dots, \dy$, it holds that $\abs{\lambda_i(\Sigma_{y+} \Sigma_{Y}^{-1})} =\abs{\lambda_i(\Sigma_y^{-1/2} \Sigma_{y+} \Sigma_{Y}^{-1/2})} \leq 1$, and thus $\rho(\Sigma_{y+} \Sigma_y^{-1}) \leq 1$. 
\end{proof}

\section{Proofs of \cref{sec: control} Results}
\subsection{Proof of \Cref{prop: LQR decay well-specified}}
\begin{lemma} \label{lemma: prop iv.1}
Let $\tilde J$, $K$, and $\calE(\cdot)$ be as defined in \cref{sec: control}. There exist constants $c \in \R$ and $\rho \in (0,1)$ such that 
    \begin{align*}
    &\lim_{N\to\infty} N\,\E\!\left[\abs{\tilde J\!\left(K(\hat G_N)\right)-\tilde J\!\left(K(G^\star)\right)} \;\middle|\; \calE(\hat G_N)\right] \\
    &=\trace \Bigl(\nabla_{\VEC (G)} K(G^\star)^T\nabla^2_{\VEC (K)} \tilde J(K^\star) \nabla_{\VEC (G)} K(G^\star) \Sigma_{\hat G}\Bigl) \\
    &+ c\rho^H
\end{align*}
where $\Sigma_{\hat G}$ is the asymptotic variance 
\begin{align*}
    \lim_{N \rightarrow \infty} N \var(\VEC (\hat G_N - G^\star)),
\end{align*}
and $K^\star = \argmin_K J(K)$.
\end{lemma}
\begin{proof}

Linearizing $\tilde J(K(G))$ around $G^\star$ gives
\begin{align} \label{eq: decay breakdown}
    &\lim_{N\to\infty} N\,\E\!\left[\abs{\tilde J\!\left(K(\hat G_N)\right)-\tilde J\!\left(K(G^\star)\right)} \;\middle|\; \calE(\hat G_N)\right] \notag \\
    &=\lim_{N \rightarrow \infty} N \E \Bigl [\Bigl |\frac{1}{2}\nabla_{\VEC (K)} \tilde J(K(G^\star))^T \notag\\
    & \cdot \nabla^2_{\VEC (G)} K(G^\star)[\VEC (\hat G_N - G^\star), \VEC (\hat G_N - G^\star) ] \notag\\
    &+ \frac{1}{2}\VEC (\hat G_N - G^\star)^T \nabla_{\VEC (G)} K(G^\star)^T\nabla^2_{\VEC (K)} \tilde J(K(G^\star)) \notag\\
    &\cdot \nabla_{\VEC (G)} K(G^\star)\VEC (\hat G_N - G^\star) \Bigl| \mid \calE(\hat G_N)\Bigl ] \notag\\
    &= \Bigl |\frac{1}{2}\nabla_{\VEC (K)} \tilde J(K(G^\star))^T V \notag\\
    &+ \frac{1}{2} \trace \Bigl(\nabla_{\VEC (G)} K(G^\star)^T\nabla^2_{\VEC (K)} \tilde J(K(G^\star)) \nabla_{\VEC (G)} K(G^\star) \Sigma_{\hat G}\Bigl) \Bigl|
\end{align}
where 
\begin{align*}
V
&=
\lim_{N\to\infty}
N \E\Bigl[
\nabla^2_{\VEC(G)} K(G^\star)
\bigl[
\VEC(\hat G_N - G^\star),
\VEC(\hat G_N - G^\star)
\bigr]
\\
&\hspace{6em}
\Bigm|\;
\calE(\hat G_N)
\Bigr].
\end{align*}

Let $K^\star = \argmin_K J(K)$. Denote 
\begin{align*}
    \Delta \tilde J \triangleq \nabla_{\VEC (K)}\tilde J(K(G^\star)) - \nabla_{\VEC (K)}\tilde J(K^\star)
\end{align*} and 
\begin{align*}
    \Delta^2 \tilde J \triangleq \nabla^2_{\VEC (K)}\tilde J(K(G^\star)) - \nabla^2_{\VEC (K)}\tilde J(K^\star).
\end{align*} The expression \eqref{eq: decay breakdown} is equal to 
\begin{align*}
    & \Bigl |\frac{1}{2}\paren{\Delta \tilde J}^\top V  \frac{1}{2} \trace \Bigl(\nabla_{\VEC (G)} K(G^\star)^\top \paren{\Delta^2 \tilde J} \nabla_{\VEC (G)} K(G^\star) \Sigma_{\hat G}\Bigl) \\
    &+ \frac{1}{2} \trace \Bigl(\nabla_{\VEC (G)} K(G^\star)^\top\nabla^2_{\VEC (K)} \tilde J(K^\star) \nabla_{\VEC (G)} K(G^\star) \Sigma_{\hat G}\Bigl)\Bigl |
\end{align*}
Let 
\begin{align*}
    &\tilde E \triangleq \frac{1}{2}\paren{\Delta \tilde J}^\top V \\
    &+ \frac{1}{2} \trace \Bigl(\nabla_{\VEC (G)} K(G^\star)^\top\paren{\Delta^2 \tilde J} \nabla_{\VEC (G)} K(G^\star) \Sigma_{\hat G}\Bigl).
\end{align*}
By Cauchy-Schwarz and von Neumann's trace inequality,  
\begin{align*}
    &\abs{ \tilde E } \leq \frac{1}{2}\norm{\Delta \tilde J}_2 \norm{V}_2 \\
    &+ \frac{1}{2} \norm{\nabla_{\VEC (G)} K(G^\star)}_2^2 \norm{\Delta^2 \tilde J}_\star\norm{\Sigma_{\hat G}}_2.
\end{align*}
By norm equivalence, there then exists a constant $C$ independent of $H$ so that 
\begin{align*}
    &\abs{ \tilde E } \leq \frac{1}{2}\norm{\Delta \tilde J}_2 \norm{V}_2 \\
    &+ C \norm{\nabla_{\VEC (G)} K(G^\star)}_2^2 \norm{\Delta^2 \tilde J}_2\norm{\Sigma_{\hat G}}_2.
\end{align*}
Note that $K(G^\star)$ converges to $K^\star$ as $H \rightarrow \infty$. By convergence of the Ricatti equation, $\norm{\Delta \tilde J}_2$ and $\norm{\Delta^2 \tilde J}_2$ decay as $\rho^H$ for some $\rho \in (0,1)$. Since each of the terms $\norm{V}_2, \norm{\nabla_{\VEC (G)} K(G^\star)}_2,$ and $\norm{\Sigma_{\hat G}}_2$ can be uniformly bounded in $H$, it follows that $\abs{ \tilde E }$ decays with $H$ as $\rho^H$.  

\end{proof}

\end{document}